\documentclass[journal,12pt,onecolumn,draftclsnofoot,]{IEEEtran}
\usepackage{stfloats}
\usepackage{cite}
\usepackage{amsmath,amssymb,amsfonts}
\usepackage{amsthm}
\usepackage{algorithmic}
\usepackage{graphicx}
\usepackage{textcomp}
\usepackage{xcolor}
\usepackage{float}
\usepackage{savesym}
\savesymbol{checkmark}
\usepackage{dingbat}
\usepackage{subfigure}
\PassOptionsToPackage{hyphens}{url}\usepackage{hyperref}
\usepackage{amsthm}
\usepackage{bm}
\usepackage{multicol}

\usepackage{fancyhdr}
 \usepackage{lastpage}
\usepackage{bm}

\newcommand{\boldPhi}{\boldsymbol{\Phi}}

\newcommand{\tr}{\text{trace}}
\newcommand{\Find}{\text{Find}}

\usepackage{cite}
\usepackage{amsmath,amssymb,amsfonts}
\usepackage{algorithmic}
\usepackage{graphicx}
\usepackage{textcomp}
\usepackage{xcolor}
\usepackage{float}
\usepackage{subfigure}

\usepackage{bm}
\usepackage{multicol}
\usepackage{graphicx}  

\usepackage{caption}
\usepackage{environ}         
\usepackage{etoolbox}        

 \usepackage{epsfig,endnotes}

\usepackage{floatpag,enumitem}
\usepackage{relsize}
 \usepackage{tikz}

\usepackage{multirow}
\usepackage{setspace}
\usepackage{mathrsfs}
\usepackage{bbm}
\usepackage{pifont}
\usepackage{amsfonts}

\usepackage{amsmath, amsthm}

\usepackage{tabularx}
\usepackage{listings}
\usepackage{graphicx}
\usepackage{amssymb,booktabs}
\usepackage{array}
\usepackage{cases}

 \usepackage{supertabular}
\usepackage{mathrsfs}

\usepackage{psfrag}
\usepackage{bbding}

\usepackage{float}

\usepackage{amsbsy}

\makeatletter
\providecommand{\leftsquigarrow}{%
  \mathrel{\mathpalette\reflect@squig\relax}%
}
\newcommand{\reflect@squig}[2]{%
  \reflectbox{$\m@th#1\rightsquigarrow$}%
}
\makeatother

\newenvironment{proof}{\noindent \textbf{Proof:}}{\hfill$\square$}
\usepackage{algorithm}
 \usepackage{algorithmic}

\makeatletter
\newcommand{\newalgname}[1]{%
  \renewcommand{\ALG@name}{#1}%
}

\def\centerhack#1{\hbox to 0pt{\hss\footnotesize #1\hss}}
\def\centerhackn#1{\hbox to 0pt{\hss #1\hss}}
\def\dchack#1{\vbox to 0pt{\vss{\hbox to 0pt{\hss#1\hss}}\vss}}

\usepackage{etoolbox}
\AtBeginEnvironment{align}{\setcounter{subeqn}{0}}
\newcounter{subeqn} %

\usetikzlibrary{positioning}

\newcounter{mysub}

\newtheorem{rem}{Remark}

\newtheorem{prop}{Proposition}
\def\BibTeX{{\rm B\kern-.05em{\sc i\kern-.025em b}\kern-.08em
    T\kern-.1667em\lower.7ex\hbox{E}\kern-.125emX}}
\begin{document}

\title{Reconfigurable Intelligent Surface Aided Power Control for Physical-Layer Broadcasting\\
}

\author{Huimei Han, Jun Zhao, Wenchao Zhai, Zehui Xiong, Dusit Niyato, {\it{Fellow, IEEE}}, Marco Di Renzo, {\it{Fellow, IEEE}},  Quoc-Viet Pham, Weidang Lu, and Kwok-Yan Lam, {\it{Senior Member, IEEE}}

\thanks{
Huimei Han  and  Weidang Lu are with  Zhejiang University of Technology, Hangzhou, Zhejiang, P.R. China.
Jun Zhao, Zehui Xiong,  Dusit Niyato and Kwok-Yan Lam are  with School of Computer Science and Engineering, Nanyang Technological University, Singapore. Wenchao Zhai is with the China Jiliang University, Hangzhou, Zhejiang, P.R. China. Marco Di Renzo is  with Laboratory of Signals and Systems of Paris-Saclay University - CNRS and CentraleSup\'{e}lec,  France.  Quoc-Viet Pham is with Research Institute of Computer, Information and Communication, Pusan National University,  South Korea.  (Emails: \{hmhan1215,\,luweid\}@zjut.edu.cn, \{junzhao,\,zxiong002,\,dniyato,\,kwokyan.lam\}@ntu.edu.sg, zhaiwenchao@cjlu.edu.cn,
 marco.direnzo@centralesupelec.fr, vietpq@pusan.ac.kr). Part of the work has been presented in
IEEE ICC 2020
~{{\cite{han2019intelligent}}} .}

}


\maketitle

 \pagestyle{plain} \thispagestyle{plain}

\begin{abstract}
Reconfigurable intelligent surface (RIS), a recently introduced technology for  future wireless communication systems, enhances the spectral and energy
efficiency by intelligently adjusting the propagation conditions between a base station (BS) and  mobile
equipments (MEs). An RIS consists of many low-cost passive reflecting elements to improve the quality of the received signal.  In this paper,  we study the problem of power control at the BS for the RIS aided
physical-layer broadcasting. Our goal is to minimize
the transmit power at the BS by jointly designing the transmit beamforming at the BS and the phase
shifts of the passive elements at the RIS.
 Furthermore, to help validate 
the proposed optimization methods,
we derive lower bounds to quantify the average transmit power at the BS as a function of the number of MEs, the number of RIS elements, and the number of antennas at the BS.
The simulation results  demonstrated that the average transmit power at the BS  is
close to the lower bound in an  RIS aided system, and is significantly lower than the average transmit power in conventional schemes
without the RIS.
\end{abstract}


\section{Introduction}


Fifth-generation (5G) communications achieve great improvement in  spectral efficiency  by utilizing various  advanced technologies,
such as massive multiple-input multiple-output (MIMO) communications,  non-orthogonal multiple access transmission, millimeter (mm)-wave  communications, and ultra-dense Heterogeneous Networks.
However, the high energy consumption of 5G communications is a critical issue~\cite{HuangTWC}.
To improve the spectral efficiency  and reduce the energy consumption simultaneously, researchers are exploring new ideas for future wireless systems beyond 5G\mbox{~\cite{RISimpleDinew,HuangAZYD19,HuangAZDY18,add1}}.  {\color{black}Among solutions, reconfigurable intelligent surface (RIS) 
can improve coverage capability, spectral efficiency, and quality of the reflected signal by controlling the phase of the incident signal in a passive way. 
Such good features make RIS
a new technology for the 6G communications, and RIS  has attracted much attention recently
~\cite{wu2018intelligentfull,wu2019beamformingICASSP,Los,RISDi,H-MIMODi,DiRe,6G}.}
An RIS is a planar array  consisting of  many reflecting and nearly passive   units, 
which intelligently and dynamically  adjusts the propagation conditions to improve the communication quality between the base station (BS) and mobile equipments (MEs). Since each RIS unit  reflects the signal in a passive way instead of transmitting/receiving signal in an active way,  the energy consumption is  low. 
{\color{red} Furthermore, the
RIS has  the features of  lightweight, low profile, and  conformal geometry, making it easy to  mount  or remove the RIS  from objects. 
If some objects (i.e., UAVs, drones, and buildings) are located between the BS and the RIS, the RIS  can be installed on a high rise building
 or in the air
  to create a line-of-sight (LoS) link between the BS and the RIS~\cite{survey,Los,IRSlocation,1-5}.
}
Indeed, passive reflecting surfaces have been  used in other communication systems, such as radar systems and satellite communication systems, but
 have received less attention in mobile wireless communication systems up until now.

{\color{black}Broadcast traffic can be used for the BS to broadcast system public control information to MEs. System information broadcasting in the communication system provides  information to facilitate MEs to establish wireless connections~\cite{broad}. Broadcast traffic can also be used for news feed, video-conference, or movie broadcast~\cite{complexity}.}
To address  the problem of power control  at the BS for physical layer
broadcasting with quality of service (QoS)  constraints  in  RIS aided networks, we propose to  employ  two alternating optimization algorithms in order to jointly design the transmit beamforming  at  the BS and the phase shifts of the reflecting units at  the  RIS.
We also present  computational  complexity analysis of the  proposed alternating optimization algorithms. 
Furthermore, to validate the performance  of the optimization algorithms,  we derive  two lower bounds for  the average transmit power at the BS.
 Simulation results  show that  the average transmit power at the BS is  close to   the lower bounds, and   is much lower  than the average transmit power of communication systems without the  RIS. Specifically, the main features and contributions of this paper can be summarized as follows.

 \begin{itemize}

\item We formulate a BS power control optimization problem for physical layer broadcasting under QoS constraints in the RIS aided network. We  propose two alternating optimization algorithms for solving it.
     Specifically, we propose an alternating optimization algorithm  based on  semidefinite relaxation (SDR) technique to obtain the minimum transmit power. To reduce  computational complexity and improve the performance,  we also propose an  alternating optimization algorithm based on  successive convex approximation (SCA) method.


\item We introduce two  lower bounds providing information on the average transmit power at the BS as a function of the number of MEs, the number of RIS units, and the number of  antennas at the BS, considering a full-rank LoS channel between the BS and the RIS. These lower bounds are employed to analyze the effectiveness of the  proposed optimization algorithms. In particular, we propose an analytical and a semi-analytical lower bounds. The first bound is easier to compute but the second one is tighter.
     Simulation results demonstrate that the average transmit power at the BS   is  close to  the semi-analytical lower bound.

\item With the aid of numerical simulations, we compare the performance of the proposed algorithms against a two-stage algorithm and conventional schemes without RISs.
\end{itemize}

As opposed to the conference version in~\cite{han2019intelligent}, we make the following extensions.
\begin{itemize}

 \item {\color{black}In contrast to the BS-to-RIS rank-one  LoS channel scenario in~\cite{han2019intelligent}, we consider a full rank LoS channel
 to benefit from the RIS in the broadcasting multi-MEs setting~\cite{Los}.}

  \item In~\cite{han2019intelligent}, we briefly introduced   an alternating optimization algorithm  based on the SDR technique to obtain the minimum transmit power. In this paper, we not only discuss the  alternating optimization algorithm  based on the SDR technique in detail, but also  propose an alternating optimization algorithm  based on the SCA method to reduce the computational complexity and improve the performance, by utilizing the characteristic of no multi-user interference in the broadcast scenario.

 \item {\color{black}In~\cite{han2019intelligent}, we introduced only   an analytical lower bound for the average transmit power with the RIS. In this paper, for benchmarking purposes, we also derive analytical lower bounds for the transmit power without the RIS and with random phase shifts at the RIS. Furthermore, to get a  tighter lower  bound, we introduce a semi-analytical lower bound.}
\end{itemize}

\textbf{Organization.}
The remainder of this paper is organized as follows. In Section~\ref{sec:Related-work}, we describe existing research works related to the present paper.
Section~\ref{sec:systemmodel}  introduces the system model and formulates the power control optimization  problem for  physical-layer broadcasting under QoS constraints. In Section~\ref{sec:Optimization}, we introduce the proposed algorithms to solve the problem. Two lower bounds for the minimum transmit power are  derived in Section~\ref{sec:analysis}. Sections~\ref{sec:simulation} and~\ref{sec:conclution} provide numerical results and conclude this paper, respectively.

\textbf{Notation.} The notations utilized throughout this paper are described in  Table~\ref{tab:notation}.

\begin{table}[htbp]
\scriptsize
\centering
 \caption{notations.}\label{tab:notation}
   \begin{tabular}{|c|c|}
\hline
\linespread{2}
Notations & Description\\
\hline
Italic letters& Scalars\\
\hline
 Boldface
lower-case & Vectors\\
\hline
Boldface
upper-case letters& Matrices\\
\hline
$(\cdot)^T$ and $(\cdot)^H$ & The transpose and conjugate transpose of a matrix \\
\hline
$\boldsymbol{D}_{i,q}$& The element in the $i^{th}$ row and $q^{th}$  column of a matrix $\boldsymbol{D}$\\
\hline
$\boldsymbol{x}_{i}$& The $i^{th}$ element of a vector $\boldsymbol{x}$\\
\hline
$\mathbb{C}$& The set of all complex numbers \\
\hline
$\mathcal{C}\mathcal{N}( \mu  , \sigma^2 )$& A circularly-symmetric complex Gaussian distribution with mean $\mu$  and variance $\sigma^2$\\
\hline
$|| \cdot||$& The Euclidean norm of a vector\\
\hline
$| \cdot |$& The cardinality of a set\\
\hline
$\text{diag}(\boldsymbol{x})$& A diagonal matrix with the element in the $i^{th}$ row and $i^{th}$  column being the $i^{th}$ element in $\boldsymbol{x}$\\
\hline
\text{arg}($\boldsymbol{x}$)& A phase vector\\
\hline
$\mathbb{E}(\cdot)$ &The expectation operator\\
\hline
$\textup{Var}(\cdot)$& The variance operator\\
\hline
$\boldsymbol{M}^{-1}$  and
$\boldsymbol{M}\succeq 0$ & The inverse  of a square  matrix $\boldsymbol{M}$ and positive \mbox{semi-definiteness}\\
\hline
$\bm{{I}}$& The identity matrix\\
\hline
$\Re(c)$, $\Im(c)$ and  $\varphi_c$& The real part, imaginary part and angle of a complex number $c$\\
\hline
\end{tabular}
\end{table}

\section{Related works}\label{sec:Related-work}

In this section, we discuss the existing research works that are related to the present paper.

 \textbf{Power control in  RIS aided communications}. Power control has been studied in several recent papers on  RIS aided communications.  Wu and Zhang  minimized the transmit power at the BS by optimizing the transmit beamforming at the
BS and the phase shifts of the RIS, subject to   signal-to-interference-plus-noise  ratio  (SINR) constraints at  MEs~\cite{wu2018intelligentfull}. They extended the study in~\cite{wu2018intelligentfull} to account for discrete phase shifts in~\cite{wu2019beamformingICASSP} and~\cite{wu2019beamforming}. Zhou~\textit{et~al.} extended~\cite{wu2018intelligentfull} to consider power control in  RIS aided communications under   imperfect channel state information (CSI)~\cite{2-7}. The problem of  power control when an RIS is used to achieve secure communications was investigated in~\cite{chu2019intelligent,feng2019secure}. Recent studies~\cite{li2019joint,fu2019intelligent} tackled transmit power minimization in  RIS aided networks with \mbox{non-orthogonal} multiple-access (NOMA).  Different from~\cite{wu2018intelligentfull,wu2019beamformingICASSP,wu2019beamforming,han2019intelligent,zhao2019optimizations} that addressed only information transfer, Wu and Zhang considered power control for simultaneous wireless information and power transfer in  RIS aided systems~\cite{wu2019joint}.

\textbf{Broadcast/multicast traffic in  RIS aided communications}.
The above mentioned papers consider the unicast setting, where the BS sends different data to different MEs. We are only aware of a few recent studies on broadcast/multicast traffic in the context of RIS aided communications. The multicast setting categorizes MEs into a number of groups, where MEs in the same group receive the same data from the BS. Zhou~\textit{et~al.} maximized the sum rate of groups in the multicast setting~\cite{Multi}. In the broadcast case, Du~\textit{et~al.} maximized the information rate under  transmit power constraints, where the system model only involves the indirect channels between the BS and MEs, and analyzed the  asymptotic  growth of the capacity~\cite{du2019multiple}. In addition, Nadeem~\textit{et~al.}
 showed  that RISs achieve better performance  when the BS-RIS LoS channel is of full rank for a multi-MEs setting~\cite{Los}.
Both~\cite{Multi} and~\cite{du2019multiple} do not consider  power control and do not  assume a full-rank LoS channel model between the BS  and the RIS, as addressed in this paper. In addition, an earlier draft~\cite{zhao2019optimizations} (not submitted) of our   paper summarized problems related to power control under QoS in the unicast, multicast, and broadcast settings, but only SDR-based algorithms were briefly presented without any analysis and simulation validation.

\textbf{Power control in traditional  communication systems without RIS}. In the absence of RIS, downlink power control under QoS for the broadcast setting was studied in the seminal work of Sidiropoulos~\textit{et~al.}~\cite{sidiropoulos2006transmit}.  The problem was also shown to be NP-hard~\cite{sidiropoulos2006transmit}. Karipidis~\textit{et~al.} extended~\cite{sidiropoulos2006transmit} to the multicast scenario~\cite{karipidis2008quality}. Power control for multicast traffic in NOMA systems was addressed in~\cite{choi2015minimum}.

 \textbf{Comparing   RIS with other technologies}.
 The main advantage of  RISs over  existing technologies (such as massive  MIMO communication~\cite{vsMIMO},  \mbox{mmWave} communication~\cite{millimeter}, amplify-and-forward relaying~\cite{risdifferDi}) is   that the RIS comprises only passive elements,  achieves low hardware cost, low energy consumption, and intelligently adjusts   the wireless environment~\cite{wu2018intelligentfull,IRSoverviewWu}.

 \textbf{RIS implementations}.
 To validate  the feasibility RIS aided systems, Xin~\textit{et~al.} implemented a  reflecting array for IEEE 802.11~ad for application to mmWave communications~\cite{IRSimplement}. Tang~\textit{et~al.} experimentally verified  channel models for different cases in  RIS aided  communication \mbox{systems}~\cite{PathmodelDi}. Examples of  RIS aided communication  prototypes were described in~\cite{RISp,RISimpleDi}.

 \textbf{Other studies of RIS}. In addition to the papers discussed above, there are many other RIS studies that address various optimization problems and emerging applications.
 {\color{black} Some machine learning based methods were proposed for the  phase shifts optimization.  Huang~\textit{et~al.} proposed a deep learning (DL) method of reconfiguring RIS online for the complex indoor environment~\cite{HuangAZYD19}. Yang~\textit{et~al.} proposed a win or
learn fast policy hill-climbing learning approach to jointly optimize the anti-jamming power allocation and reflecting beamforming~\cite{IRSML1}.
}
We refer interested readers to a recent surveys~\cite{RISimpleDinew,gong2019towards} and the references therein.

\section{System model and  problem definition}\label{sec:systemmodel}

\subsection{System model}
\begin{figure}[!t]
  \centering
\includegraphics[scale=0.45]{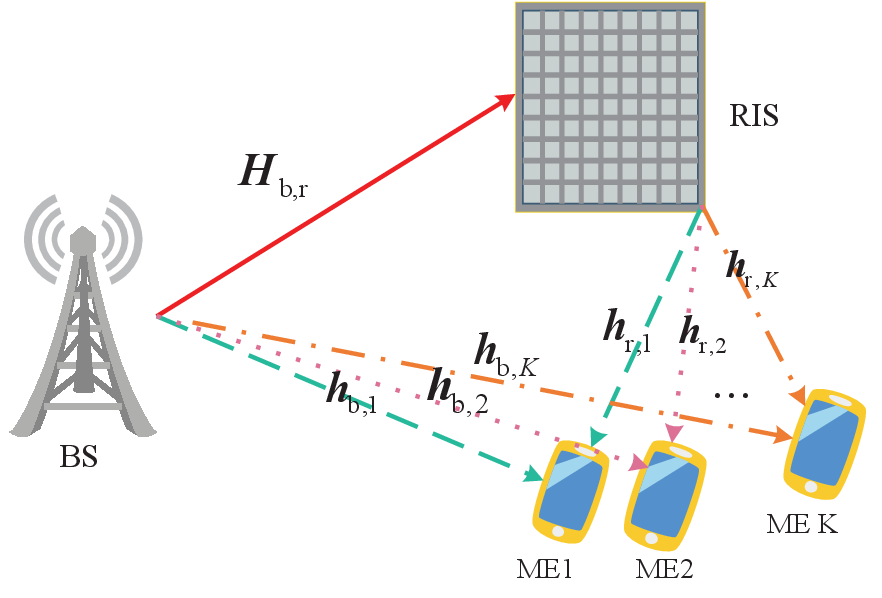}\\
  \caption{An  RIS aided communication system consisting of a base station (BS), multiple mobile equipments (MEs), and an RIS comprising many RIS units.
  \vspace{-10pt}}\label{systemmodel}
\end{figure}

We consider an  RIS aided multi-input single-output (MISO) communication system in a broadcast setting, where there is a  BS with $M$ antennas and an RIS with $N$ reflecting units, and  $K$ single-antenna MEs, as shown in Fig.~\ref{systemmodel}. The BS uses linear  transmit precoding (i.e., beamforming) that is denoted by the vector $\bm{w}\in \mathbb{C}^{M\times1}$. So, the transmitted signal at the BS is $\boldsymbol{x}= \bm{w}{s}$, where ${s}$ is  the broadcasted data signal. The signal $\boldsymbol{x}$   reachs each ME via  both indirect and direct  channels, and the received signal at each ME is the sum of the signals from these two channels. More specifically, for the indirect channel, the transmitted signal $\boldsymbol{x}$  travels from the BS  to the RIS, it is then reflected by the RIS, and  finally travels from the RIS to these $K$ MEs. For the direct channel, the  transmitted signal $\boldsymbol{x}$  travels from  the BS to these $K$ MEs directly.

Let $\boldPhi= \text{diag}(\beta_1 e^{j \theta_1}, \ldots, \beta_n e^{j \theta_n} , \ldots, \beta_N e^{j \theta_N})$ denote  the reflection coefficients matrix at the RIS, where  $\beta_n$ and $\theta_n$ denote the amplitude and  phase shift, respectively, and $j$ is the unit imaginary number. Similar to other recent studies~\cite{wu2018intelligentfull,wu2019beamformingICASSP,Los}, we assume that the RIS changes only the phase of the reflected signal, i.e., $\theta_n \in [0,2\pi)$ and $\beta_n=1${\color{black}\footnote{{\color{black}
Since each element of the RIS is usually designed to maximize the
signal reflection, we set the  amplitudes of all RIS elements to 1 for simplicity in our papers~\cite{wu2018intelligentfull}. When the reflection amplitude depends
on the phase shift, the reflection design proposed in this paper is no longer optimal in general and will cause
performance degradation~\cite{add3}, and we will do some research on  the coupling of reflection amplitude and phase shift   in the future. 
}}}. Let $\bm{H}_{\text{b},\text{r}} \in \mathbb{C}^{N \times M} $, $\boldsymbol{h}_{\text{r},i}^H \in \mathbb{C}^{1 \times N} $,  and $\boldsymbol{h}_{\text{b},i}^H\in \mathbb{C}^{1 \times M} $ be the BS-RIS channel, RIS-$i^{th}$ ME channel, and BS-$i^{th}$ ME channel, respectively (the subscripts b, r, and $i$ represent the BS, the RIS, and the $i^{th}$ ME, respectively). Then, the received signal at ME $i$ is given by
\begin{equation}\label{receive-data}
\setlength{\abovedisplayskip}{2pt plus 1pt minus 1pt}
\setlength{\belowdisplayskip}{2pt plus 1pt minus 1pt}
\setlength\abovedisplayshortskip{2pt plus 1pt minus 1pt}
\setlength\belowdisplayshortskip{2pt plus 1pt minus 1pt}
{y_i} = (\boldsymbol{h}_{\text{r},i}^H\boldPhi {\bm{H}_{\text{b},\text{r}}} + \boldsymbol{h}_{\text{b},i}^H)\bm{w}{s}  + {n_i},~i = 1,\ldots,K,
\end{equation}
where  ${n_i}\sim \mathcal{C}\mathcal{N}(0, \sigma^2_i) $ denotes the  additive white  Gaussian noise at the $i^{th}$ ME.

We assume that the broadcasted  data signal ${s}$ is  normalized to have unit power. Then, the \mbox{signal-to-noise}  ratio  (SNR) at the $i^{th}$ ME  can be written as
\begin{equation}\label{SINR}
\text{SNR}_i   = \frac{|\bm{h}^H_{i}(\boldPhi)\bm{w}|^2}{ \sigma^2_i},~i = 1,\ldots,K,
\end{equation}
where $\boldsymbol{h}_i^H(\boldPhi)= \bm{h}^H_{\text{r},i}\boldPhi \bm{H}_{\text{b},\text{r}}+\bm{h}^H_{\text{b},i}$ denotes the overall downlink channel from the BS to  ME $i$.
\subsection{Problem definition}
Our research problem on power control is to minimize  the  transmitted power at the BS for  broadcasting under QoS constraints. Note that the transmitted power at the BS is $\|\bm{w}\|^2$, and  the QoS of the $i^{th}$ ME can  be characterized by its SNR, which is required to be at least  $\gamma_i$.
Then, the optimization problem can be formulated as follows. 
\begin{subequations} \label{eq-Prob-P1}
\setlength{\abovedisplayskip}{2pt plus 1pt minus 1pt}
\setlength{\belowdisplayskip}{2pt plus 1pt minus 1pt}
\setlength\abovedisplayshortskip{2pt plus 1pt minus 1pt}
\setlength\belowdisplayshortskip{2pt plus 1pt minus 1pt}
\begin{alignat}{1}
 \text{(P1):~}
\min_{\bm{w}, \boldPhi}~ &\|\bm{w}\|^2 \label{OptProb-power-broadcast-obj}   \\
~\mathrm{s.t.}~&\frac{|\bm{h}^H_{i}(\boldPhi)\bm{w}  |^2}{\sigma^2_i}\geq \gamma_i, ~i =1,\ldots,K, \label{OptProb-power-broadcast-SINR-constraint11}\\
&0\le {\theta _n} < 2\pi , \ n=1,\ldots,N. \label{OptProb-power-broadcast-Phi-constraint}
\end{alignat}
\end{subequations}

{\color{black}Let us define the normalized channel vector $ {{\hat{\bm{h}}_{i}(\boldPhi)}}:=\frac{\bm{h}_{i}(\boldPhi)}{\sigma_i\sqrt{\gamma_i}}$~\cite{complexity}. Then, the constraint (\ref{OptProb-power-broadcast-SINR-constraint11}) can be transformed as ${|{{\hat{\bm{h}}_{i}(\boldPhi)}}\bm{w}|^2}\geq 1$.
 This indicates that, though $\gamma_i$ is different for different MEs, $\gamma_i$ can be coupled into the channel vector of each ME, thereby all MEs having the same threshold. In our paper, similar to other recent studies~\cite{wu2018intelligentfull,wu2019beamformingICASSP,samegama},  we
consider that all MEs have the same threshold  and  noise variance  for simplicity (i.e., $\gamma_i=\gamma$, $\sigma_i^2=\sigma^2,~\forall~i =1,\ldots,K$).}


 \section{Alternating optimization algorithms } \label{sec:Optimization}
Problem (P1) is NP-hard since even the setup without RIS is NP-hard according to~\cite{sidiropoulos2006transmit}. To  solve  Problem~(P1), we utilize  the alternating optimization to perform the following operations iteratively: optimize $\bm{w}$ given $\boldPhi$, and optimize $\boldPhi$ given $\bm{w}$. Under this framework, we  propose two algorithms based on
SDR and
SCA, respectively.

\subsection{Alternating optimization algorithm  based on  SDR  } \label{sec:Optimization-SDR}
\begin{algorithm}[h]
\small
\caption{Alternating optimization based on SDR to find $\bm{w}$ and $\boldsymbol{\Phi}$ for Problem~(P1).} \label{Alg-Problem-P1}
\begin{algorithmic}[1]
\setlength{\belowcaptionskip}{-2cm}
\setlength{\abovecaptionskip}{-2cm}
\STATE Initialize  $\boldsymbol{\Phi}$ as $\boldsymbol{\Phi}^{(0)}:= \text{diag}( e^{j \theta_1^{(0)}}, \ldots,  e^{j \theta_N^{(0)}})$, where $\theta_n^{(0)}$  $(n =1,2,\ldots, N)$ is chosen uniformly at random from $[0, 2\pi)$;
\STATE Initialize the iteration number $q \leftarrow 1$;
\WHILE{1}
\item[] \COMMENT{\textit{Comment: Optimizing $\bm{w}$ given $\boldPhi$:}}
\STATE {Given \hspace{-1pt}$\boldPhi$ as \hspace{-1pt}$\boldsymbol{\Phi}^{(q-1)}$, solve Problem (P3) to obtain \hspace{-1pt}$\bm{w}^{(q)}$\hspace{-1pt};\hspace{-1pt}}
\STATE \mbox{Compute the object function value $P_t^{(q)} \leftarrow \|\bm{w}^{(q)}\|^2$};
\IF{$1-\frac{P_t^{(q)}}{P_t^{(q-1)}}\le \varepsilon $}
\STATE \textbf{break}; 
\COMMENT{\textit{Comment: $\varepsilon$  controls the number of executed iterations before termination. The algorithm terminates if the relative difference between the transmit power obtained during the $q^{th}$ iteration and the $(q-1)^{th}$ iteration is no greater than $\varepsilon$.}} \label{Alg-Problem-P1-break1}
\ENDIF
\item[] \COMMENT{\textit{Comment: Finding $\boldPhi$ given $\bm{w}$:}}
\STATE Given $\bm{w}$ as $\bm{w}^{(q)}$, solve Problem~(P7) to obtain $\boldPhi^{(q)}$;
\IF{Problem~(P7) is infeasible}
\STATE \textbf{break};  \label{Alg-Problem-P1-break2}
\ENDIF
\ENDWHILE
\end{algorithmic}
\end{algorithm}
 The SDR-based alternating optimization is shown in Algorithm \ref{Alg-Problem-P1}.  The details of the $q^{th}$ iteration are described as follows.

\textbf{Optimizing $\bm{w}$ given $\boldPhi^{(q-1)}$.}
Given $\boldPhi^{(q-1)}$, 
Problem~(P1) reduces to a conventional power control problem under QoS in the downlink broadcast channel in the absence of the RIS:
\begin{subequations} \label{eq-Prob-P1a}
\setlength{\abovedisplayskip}{2pt plus 1pt minus 1pt}
\setlength{\belowdisplayskip}{2pt plus 1pt minus 1pt}
\setlength\abovedisplayshortskip{2pt plus 1pt minus 1pt}
\setlength\belowdisplayshortskip{2pt plus 1pt minus 1pt}
\begin{alignat}{2}
 \text{(P2):~}
\min_{\bm{w}}~ &\|\bm{w}\|^2 \label{OptProb-power-broadcast-obj}   \\
~\mathrm{s.t.}~&\frac{|\bm{h}^H_{i}(\boldPhi^{(q-1)})\bm{w}  |^2}{\sigma^2}\geq \gamma, ~i =1,\ldots,K. \label{OptProb-power-broadcast-SINR-constraintp2}
\end{alignat}
\end{subequations}

Problem (P2) is \mbox{non-convex} because of the \mbox{non-convex} constraints  ~({\ref{OptProb-power-broadcast-SINR-constraintp2}}). Define  $\bm{X}:=\bm{w}\bm{w}^H$ and $\bm{H}_{i}(\boldPhi^{(q-1)}):=\bm{h}_{i}(\boldPhi^{(q-1)})\bm{h}_{i}(\boldPhi^{(q-1)})^H$. Then,  Problem (P2) can be rewritten as~\cite{SDP}
\begin{subequations}
\begin{alignat}{1}
\setlength{\abovedisplayskip}{4pt plus 1pt minus 1pt}
\setlength{\belowdisplayskip}{4pt plus 1pt minus 1pt}
\setlength\abovedisplayshortskip{4pt plus 1pt minus 1pt}
\setlength\belowdisplayshortskip{4pt plus 1pt minus 1pt}
 \text{(P3):~}
\min_{\bm{X}} ~~~ &\tr(\bm{X})  \\
~\mathrm{s.t.}~&\tr(\bm{X} \bm{H}_{i}(\boldPhi^{q-1}))\geq \gamma\sigma^2, ~ i =1,\ldots,K, \label{OptProb-power-broadcast-SINR-constraint}\\
& \boldsymbol{\bm{X}} \succeq 0,~ \text{rank}(\boldsymbol{\bm{X}})= 1.\label{OptProb-power-broadcast-rankone-constraint}
\end{alignat}
\end{subequations}
We employ  SDR to drop the \mbox{non-convex} rank-one constraint in ({\ref{OptProb-power-broadcast-rankone-constraint}}). Therefore, Problem (P3) reduces to a semi-definite program  (SDP), and  we can utilize conventional optimization software (e.g., CVX~\cite{CVX}) to solve it. Generally, SDR may not produce a rank-one solution to Problem (P3). Thus, once $\bm{X}$ is available, the Gaussian randomization method~\cite{Gaussianrandom} is usually applied to obtain rank-one
 solutions to Problem (P3). Specifically, by utilizing the eigenvalue decomposition, $\bm{X}$ can be rewritten as $\bm{X}=\bm{U}\bm{\Lambda}\bm{U^{H}}$  where $\bm{U}$ and $\bm{\Lambda}$ are a unitary matrix and a diagonal
matrix, respectively. Then,  a suboptimal solution to Problem (P3) is  given by $\bm{X}_s=\bm{U}\bm{\Lambda}^{1/2}\bm{r}$, where $\bm{r}$ is a random vector whose distribution is  $\mathcal{C}\mathcal{N}( 0, \bm{I}_M )$ with $\bm{I}_M$ being the identity matrix.
When utilizing the  Gaussian randomization, we can obtain multiple candidate solutions to  Problem (P3), {\color{red} and we  regard $\bm{w^{(q-1)}}$ as a candidate solution to  Problem (P3) to ensure the convergence of the proposed SDR-based algorithm as shown in Proposition 1.} Then, we select the candidate solution with the minimum transmit  power at the BS as the value of $\bm{w}$ during the $q^{th}$ iteration,  denoted by $\bm{w}^{(q)}$. 

\textbf{Finding $\boldPhi$ given $\bm{w}^{(q)}$.} Given $\bm{w}^{(q)}$, Problem~(P1) boils down to the following feasibility check problem of finding $\boldPhi$:
\begin{subequations} \label{eq-Prob-P1c}
\setlength{\abovedisplayskip}{2pt plus 1pt minus 1pt}
\setlength{\belowdisplayskip}{2pt plus 1pt minus 1pt}
\setlength\abovedisplayshortskip{2pt plus 1pt minus 1pt}
\setlength\belowdisplayshortskip{2pt plus 1pt minus 1pt}
\begin{alignat}{2}
\text{(P4)}: \Find~~~& \boldPhi \\
\mathrm{s.t.}~~~~~&   \frac{|\bm{h}^H_{i}(\boldPhi)\bm{w}^{(q)} |^2}{\sigma^2}\geq \gamma, ~ i =1,\ldots,K, \label{eq:prob-P2-SINR-constraint}\\
 &0\le {\theta _n} < 2\pi , \ n=1,\ldots,N. \label{eq:prob-P1a-rank-constraint}
\end{alignat}
\end{subequations}
Define  $\bm{\phi}= [ e^{j \theta_1}, \ldots,  e^{j \theta_N}]^H$, $\bm{a}_{i}=\text{diag}(\bm{h}^H_{\text{r},i})\bm{H}_{\text{b},\text{r}}\bm{w}^{(q)}$, and $b_{i}=\bm{h}^H_{\text{b},i}\bm{w}^{(q)}.$  Then, Problem (P4) can be rewritten as
\begin{subequations} \label{eq-Prob-P1d}
\setlength{\abovedisplayskip}{2pt plus 1pt minus 1pt}
\setlength{\belowdisplayskip}{2pt plus 1pt minus 1pt}
\setlength\abovedisplayshortskip{2pt plus 1pt minus 1pt}
\setlength\belowdisplayshortskip{2pt plus 1pt minus 1pt}
\begin{alignat}{2}
\text{(P5)}: \Find~~~& \bm{\phi} \\
\mathrm{s.t.}~~~~~&    \begin{bmatrix}
\bm{\phi}^H ,~
1
\end{bmatrix} \bm{A}_i \begin{bmatrix}
\bm{\phi} \\
1
\end{bmatrix} + {b}_{i}{b}_{i}^H  \geq \gamma\sigma^2,~ i =1,\ldots,K, \label{eq:prob-quadratic-constraint}  \\
 &|{\phi}_n|=1 , \ n=1,\ldots,N, \label{eq:prob-P1d-rank-constraint}
\end{alignat}
\end{subequations}
where $ \bm{A}_i= \begin{bmatrix}
\bm{a}_{i}\bm{a}_{i}^H, & \bm{a}_{i}{b}_{i}^H \\
{b}_{i}\bm{a}_{i}^H, & 0
\end{bmatrix}, i =1,\ldots,K.$
Note that, since the constraints in~(\ref{eq:prob-P1d-rank-constraint}) are \mbox{non-convex},
Problem (P5) is a \mbox{non-convex} optimization problem.  Let us introduce an auxiliary variable
$t$ satisfying $|t|=1$, and define $\bm{v}:= t  \begin{bmatrix}
\bm{\phi}  \\
1
\end{bmatrix}=  \begin{bmatrix}
\bm{\phi}t \\
t
\end{bmatrix},  \boldsymbol{V}:= \boldsymbol{v}  \boldsymbol{v}^H.$
 Problem~(P5)   can  then be transformed into the following problem
\begin{subequations} \label{eq-Prob-P6}
\setlength{\abovedisplayskip}{2pt plus 1pt minus 1pt}
\setlength{\belowdisplayskip}{2pt plus 1pt minus 1pt}
\setlength\abovedisplayshortskip{2pt plus 1pt minus 1pt}
\setlength\belowdisplayshortskip{2pt plus 1pt minus 1pt}
\begin{alignat}{2}
\text{(P6)}: \Find~~~& {\boldsymbol{V}} \\
\mathrm{s.t.}~~~~~&  \tr(\bm{A}_i\boldsymbol{V} )  + {b}_{i} {b}_{i}^H \geq   \gamma \sigma^2,~ i =1,\ldots,K, \label{eq:prob-P1e-prime-SINR-constraint}\\
 &\boldsymbol{V} \succeq 0,~\text{rank}(\boldsymbol{V})=1, \boldsymbol{V}_{n,n}  = 1, ~  n=1,\ldots,N+1.
\end{alignat}
\end{subequations}

Similar  to Problem (P3), SDR is utilized to drop the \mbox{non-convex} rank-one   constraint for $\boldsymbol{V}$.
To accelerate the  optimization process, the variables $\alpha_i~(i=1,\ldots,K)$~are further introduced~\cite{wu2018intelligentfull}:
\begin{subequations} \label{eq-Prob-P7}
\setlength{\abovedisplayskip}{2pt plus 1pt minus 1pt}
\setlength{\belowdisplayskip}{2pt plus 1pt minus 1pt}
\setlength\abovedisplayshortskip{2pt plus 1pt minus 1pt}
\setlength\belowdisplayshortskip{2pt plus 1pt minus 1pt}
\begin{alignat}{2}
\text{(P7)}:  \max_{\boldsymbol{V},\boldsymbol{\alpha}} \sum\limits_{i=1}^K   &\alpha_i\\
\mathrm{s.t.}~~~~~&   \tr(\bm{A}_i\boldsymbol{V} )  + {b}_{i} {b}_{i}^H \geq \alpha_i + \gamma \sigma^2,~ i =1,\ldots,K, \label{eq:prob-P1e-prime-SINR-constraint}\\
 &\boldsymbol{V} \succeq 0, \boldsymbol{V}_{n,n}  = 1, \alpha_i \geq 0, ~n=1,\ldots,N+1,~i=1,\ldots,K.
\end{alignat}
\end{subequations}
Similar to Problem (P3),  we can utilize the CVX software~\cite{CVX} to solve Problem (P7). Generally, the SDR may not produce the rank-one solution to Problem (P7). Thus, once $\bm{V}$ is available, the Gaussian randomization method~\cite{Gaussianrandom} is again applied to obtain many candidate rank-one solutions to  Problem (P7), which are denoted by $[\boldPhi^{(q)}_1,\ldots, \boldPhi^{(q)}_c]$ where $c$ is  the number of candidate solutions.
In particular, to select one from these $c$ candidate solutions as the value of $\boldPhi$ during the $q^{th}$ iteration, which is denoted by  $\boldPhi^{{(q)}}$, we use the following procedure.

Let $P_t=\|\bm{\bm{w}}\|^2$ denote  the transmit power.  Given $\boldPhi$, Problem~(P2)~is rewritten as
\begin{equation}\label{min-power}
\setlength{\abovedisplayskip}{2pt plus 1pt minus 1pt}
\setlength{\belowdisplayskip}{2pt plus 1pt minus 1pt}
\setlength\abovedisplayshortskip{2pt plus 1pt minus 1pt}
\setlength\belowdisplayshortskip{2pt plus 1pt minus 1pt}
\begin{aligned}
&\min_{ \overline{\bm{w}}, P_t} P_t\\
&\mathrm{s.t.} \quad \frac{{P_t| \bm{h}_i^{H}(\boldPhi) \overline{\bm{w}} |}^2}{{\sigma}^2 }\geq \gamma, \quad \forall i\in\{1,2,\ldots,K\}.
\end{aligned}
\end{equation}
We define $f:={\min ({  {| \bm{h}_1^{H}(\boldPhi) \overline{\bm{w}} |^2},\ldots, {| \bm{h}_K^{H}(\boldPhi) \overline{\bm{w}}|^2} })},$
 where $\overline{\bm{w}}=\frac{\bm{w}}{\|\bm{w}\|}$ denotes the transmit beamforming direction  at the BS. By direct inspection, the minimum value of $P_t$ is
\begin{equation}\label{proof-inc}
P_t= \frac{\gamma\sigma^2}    {\min (  {| \bm{h}_1^{H}(\boldPhi) \overline{\bm{w}} |^2},\ldots, {| \bm{h}_K^{H}(\boldPhi) \overline{\bm{w}} |^2} )}=\frac{\gamma\sigma^2}{f}.
\end{equation}
 We can see from (\ref{proof-inc}) that, given $\overline{\bm{w}}$, the larger the value of $f$, the smaller the value of $P_t$. Therefore, we select the candidate solution  of $\boldPhi$ that maximizes $f$ as  $\boldPhi^{(q)}$. The maximum value of $f$ is denoted by $f^{(q)}_{\text{o}\boldPhi}$ (the subscript ``\text{o}'' denotes optimized). In order to ensure that the objective value in Problem (P2) is \mbox{non-increasing} over the iterations (as detailed in Proposition 1),  $f^{(q)}_{\text{o}\boldPhi}$ needs to satisfy the condition  $f^{(q)}_{\text{o}\boldPhi} \ge f^{(q)}_{\text{o}\bm{w}}$, where $f^{(q)}_{\text{o}\bm{w}}$ is the value of $f$ after optimizing $\bm{w}$ given $\boldPhi^{(q-1)}$. $f^{(q)}_{\text{o}\bm{w}}$ is obtained by replacing $\boldPhi$ and $\overline{\bm{w}}$  in  (\ref{proof-inc})  with $\boldPhi^{(q-1)}$  and $\overline{\bm{w}}^{(q)}=\frac{\bm{w}^{(q)}}{\|\bm{w}^{(q)}\|}$, respectively. If $f^{(q)}_{\text{o}\boldPhi} < f^{(q)}_{\text{o}\bm{w}}$, the iteration process ends.


{\color{black}\begin{prop}
The proposed SDR-based  alternating algorithm is convergent.
\end{prop}
\begin{proof}
The convergence of the proposed SDR-based  alternating algorithm is guaranteed by the
following two facts. First, the objective value in Problem (P2)
 is non-increasing over iterations. More specifically, if the value of $f$  after optimizing $\bm{w}$ given $\boldPhi$ is \mbox{non-decreasing} over the iterations, then $P_t$ is \mbox{non-increasing} over the iterations. That is,  if $f^{(q+1)}_{\text{o}\bm{w}} \ge f^{(q)}_{\text{o}\bm{w}}$,  we have $P_t^{(q+1)} \le P_t^{(q)}$. The rule of selecting  $\boldPhi^{(q)}$  ensures  $f^{{(q)}}_{\text{o}\boldPhi} \ge f^{(q)}_{\text{o}\bm{w}}$, {\color{red} which also makes $\bm{w^{(q)}}$  a solution to  Problem (P3) in the $(q+1)^{th}$ iteration.  We  regard  $\bm{w^{(q)}}$ as one of the candidate solutions to  Problem (P3), and select the candidate solution with the minimum transmit power at the BS as $ \bm{w^{(q+1)}}$ to ensure that $\bm{w}^{(q+1)}$ does not decrease the transmit power. Then, we have $f^{(q+1)}_{\text{o}\bm{w}} \ge f^{(q)}_{\text{o} \boldPhi}$ based on (\ref{proof-inc}).}
Hence, we have $f^{(q+1)}_{\text{o}\bm{w}} \ge f^{(q)}_{\text{o} \boldPhi}\ge f^{(q)}_{\text{o}\bm{w}} $ and $P_t^{(q+1)} \le P^{(q)}_t$, which means that $P_t$ is \mbox{non-increasing} over the iterations, i.e., the objective value in Problem (P2) is \mbox{non-increasing} over the iterations.
 Second, the optimal
value  is bounded from below due to the SNR
constraints for Problem (P2). Therefore, the proposed SDR-based  alternating algorithm is 
convergent.
\end{proof}}

\subsection{Alternating optimization algorithm based on  SCA} \label{sec1:Optimization2}
Problem (P1) is NP-hard, and  we have utilized the SDR-based alternating optimization algorithm to solve this problem in Section \ref{sec:Optimization-SDR}.  However,  SDR  causes performance loss, and the complexity of solving Problems (P2) and (P4)  is high (a detailed discussion is available in Section~\ref{sec1:Optimization2}).  To reduce the computational complexity and improve the performance, we propose an alternating optimization algorithm based on SCA to solve  Problem (P1), which is shown in Algorithm \ref{Alg-Problem-P1-SCA}. Specifically, we employ the SCA method to solve Problem (P2) and to reduce the computational complexity, and introduce  a variable $g$  to solve  Problem (P4) in order  to maximize the value of $\min (  {| \bm{h}_1^{H}(\boldPhi) {\bm{w}} |^2},\ldots, {| \bm{h}_K^{H}(\boldPhi) {\bm{w}} |^2} )$ when optimizing $\boldPhi$ given $\bm{w}$.

 \begin{algorithm}[!h]
 \small
\caption{Alternating optimization based on SCA to find $\bm{w}$ and $\boldsymbol{\Phi}$ for Problem~(P1).} \label{Alg-Problem-P1-SCA}
\begin{algorithmic}[1]
\setlength{\belowcaptionskip}{-0.7cm}
\setlength{\abovecaptionskip}{-0.1cm}
\STATE Initialize  $\boldsymbol{\Phi}$ as $\boldsymbol{\Phi}^{(0)}:= \text{diag}( e^{j \theta_1^{(0)}}, \ldots,  e^{j \theta_N^{(0)}})$, where $\theta_n^{(0)}$  $(n =1,2,\ldots, N)$ is chosen uniformly at random from $[0, 2\pi)$;
\STATE Initialize the iteration number $q \leftarrow 1$;
\WHILE{1}
\item[] \COMMENT{\textit{Comment: Optimizing $\bm{w}$ given $\boldPhi$:}}
\STATE {Given \hspace{-1pt}$\boldPhi$ as \hspace{-1pt}$\boldsymbol{\Phi}^{(q-1)}$, solve Problem $\text{(P8)}$ to obtain \hspace{-1pt}$\bm{w}^{(q)}$\hspace{-1pt};\hspace{-1pt}}
\STATE \mbox{Compute the object function value $P_t^{(q)} \leftarrow \|\bm{w}^{(q)}\|^2$};
\IF{$1-\frac{P_t^{(q)}}{P_t^{(q-1)}}\le \varepsilon $}
\STATE \textbf{break}; 
\COMMENT{\textit{Comment: $\varepsilon$  controls the number of executed iterations before termination. The algorithm terminates if the relative difference between the transmit power obtained during the $q^{th}$ iteration and the $(q-1)^{th}$ iteration is no greater than $\varepsilon$.}} \label{Alg-Problem-P1-break1}
\ENDIF
\item[] \COMMENT{\textit{Comment: Finding $\boldPhi$ given $\bm{w}$:}}
\STATE Given $\bm{w}$ as $\bm{w}^{(q)}$, solve Problem~(P9) to obtain $\boldPhi^{(q)}$;
\IF{Problem~(P9) is infeasible}
\STATE \textbf{break};  \label{Alg-Problem-P1-break2}
\ENDIF
\ENDWHILE
\end{algorithmic}
\end{algorithm}
\textbf{Optimizing $\bm{w}$ given $\boldPhi^{(q-1)}$.}
Given $\boldPhi^{(q-1)}$ obtained during the $(q-1)^{th}$ iteration,  Problem~(P1) becomes the conventional power control Problem~(P2).  Problem (P2) is \mbox{non-convex} because of the \mbox{non-convex} constraints   in  Inequalities~({\ref{OptProb-power-broadcast-SINR-constraintp2}}), and Section~\ref{sec:Optimization-SDR} utilizes   SDR to solve Problem~(P2). To reduce the computation complexity, we utilize the SCA-based method in~\cite{SLA}  to solve Problem~(P2). Specifically, Problem~(P2) is equivalent to
\begin{subequations}
\setlength{\abovedisplayskip}{2pt plus 1pt minus 1pt}
\setlength{\belowdisplayskip}{2pt plus 1pt minus 1pt}
\setlength\abovedisplayshortskip{2pt plus 1pt minus 1pt}
\setlength\belowdisplayshortskip{2pt plus  1pt minus 1pt}
\begin{alignat}{1}
 \text{(P8):~}
&\mathop {\min }\limits_{\bm{w},\{ {x_i},{y_i},\forall i\} }  ~~~ \|\bm{w}\|^2  \\
~\mathrm{s.t.} \ \ \ &x_i^2 + y_i^2 \geq \gamma\sigma^2, ~ i =1,\ldots,K, \label{OptProb-power-broadcast-SINR-constraintp8}\\
&{x_i} = \Re (\bm{h}^H_{i}(\boldPhi^{(j-1)})\bm{w}),{\rm{   }}{y_i} = \Im (\bm{h}^H_{i}(\boldPhi^{(j-1)})\bm{w}),~ i =1,\ldots,K,
\end{alignat}
\end{subequations}
where the set of constraints in Inequalities~(\ref{OptProb-power-broadcast-SINR-constraintp8}) are still \mbox{non-convex}. To tackle the non-convexity, we employ the SCA method where  $\bm{w}$ is obtained iteratively. Specifically,
let us define $\bm{r_i} := {({x_i},{y_i})^T}$. Based on the SCA method~\cite{SLA}, during the ${d}^{th}$ iteration, the left-hand side of the constraints in Inequalities~(\ref{OptProb-power-broadcast-SINR-constraintp8})  can be written as
\begin{equation}\label{sca}
\setlength{\abovedisplayskip}{2pt plus 1pt minus 1pt}
\setlength{\belowdisplayskip}{2pt plus 1pt minus 1pt}
\setlength\abovedisplayshortskip{2pt plus 1pt minus 1pt}
\setlength\belowdisplayshortskip{2pt plus  1pt minus 1pt}
x_i^2 + y_i^2 = \bm{r_i}^T\bm{r_i} \geq {\left\| {\bm{p_i}^{(d)}} \right\|^2} + 2\sum\limits_{b = 1}^2 {p_{i,b}^{(d)}} ({r_{i,b}} - p_{i,b}^{(d)}),~ i =1,\ldots,K,
\end{equation}
where $\bm{p_i}$ is  a parameter vector which is  updated as  ${\bm{p_i}^{(d+1)}}={\bm{r_i}^{(d)}}$ during the ${(d+1)}^{th}$ iteration,  and $p_{i,b}$ and ${r_{i,b}} $ stand for the $b^{th}$ component of vector $\bm{p_i}$ and  $\bm{r_i} $, respectively.  During each step of the iterative procedure, the convexity of  $\bm{r_i}^T\bm{r_i}$ and the first order Taylor approximation ensures that the right-hand side bounds the left-hand
side from below in ({\ref{sca}})~\cite{SLA}.  During the ${d}^{th}$ iteration,  Problem (P8) can be written as~\cite{SLA}
\begin{subequations}
\setlength{\abovedisplayskip}{2pt plus 1pt minus 1pt}
\setlength{\belowdisplayskip}{2pt plus 1pt minus 1pt}
\setlength\abovedisplayshortskip{2pt plus 1pt minus 1pt}
\setlength\belowdisplayshortskip{2pt plus  1pt minus 1pt}
\begin{alignat}{1}
 {\text{(P}}{{\text{8}}^{{{'}}}}{\text{)}}:~
&\mathop {\min }\limits_{\bm{w},\{ {x_i},{y_i},\forall i\} }  ~~~ \|\bm{w}\|^2  \\
~\mathrm{s.t.} \ \ \    & {\left\| {\bm{p_i}^{(d)}} \right\|^2} + 2\sum\limits_{b = 1}^2 {p_{i,b}^{(d)}} ({r_{i,b}} - p_{i,b}^{(d)})  \geq \gamma\sigma^2, ~ i =1,\ldots,K, \label{OptProb-power-broadcast-SINR-constraint}\\
&{x_i} = \Re (\bm{h}^H_{i}(\boldPhi^{(j-1)})\bm{w}),{\rm{   }}{y_i} = \Im (\bm{h}^H_{i}(\boldPhi^{(j-1)})\bm{w}),~ i =1,\ldots,K,
\end{alignat}
\end{subequations}
where  $\bm{p_i}$ is initialized with a value that is  in the feasible set of   Problem  $\text{(P8}^{'})$. The iterative procedure for solving Problem  (P8) is outlined in Algorithm~\ref{Alg-Problem-P3}.

\begin{algorithm}[!h]
\small
\caption{SCA method to find $\bm{w}$ for $\text{(P8)}$.} \label{Alg-Problem-P3}
\begin{algorithmic}[1]
\setlength{\belowcaptionskip}{-0.7cm}
\setlength{\abovecaptionskip}{-0.1cm}
\STATE Initialize  $\bm{p_i}$ as $\bm{p_i^{(0)}}$: Randomly generate  $\bm{p_i^{(0)}}$ that belongs to the feasible set of   Problem  ${\text{(P}}{{\text{8}}^{'}}{\text{)}}$;
\STATE Initialize the iteration number $d \leftarrow 1$;
\WHILE{1}
\STATE Solve ${\text{(P}}{{\text{8}}^{{'}}}{\text{)}}$ by utilizing convex optimization software (e.g., CVX~\cite{CVX});
\STATE Set $\bm{p_i^{(d+1)}}=\bm{r_{i}^{(d)}}$ and update $d=d+1$;
\IF{ Convergent or reach the required number of iteration}
\STATE \textbf{break};  \label{Alg-Problem-P1-break2}
\ENDIF
\ENDWHILE
\end{algorithmic}
\end{algorithm}

\textbf{Finding $\boldPhi$ given $\bm{w}^{(q)}$.} As shown in Section~\ref{sec:Optimization-SDR}, by introducing an auxiliary variable $t$, Problem (P4)  is converted into  Problem~(P6). Furthermore, (\ref{proof-inc}) shows that minimizing the transmit power is equivalent to  maximizing~$\min (  {| \bm{h}_1^{H}(\boldPhi) \overline{\bm{w}} |^2},\ldots, {| \bm{h}_K^{H}(\boldPhi) \overline{\bm{w}} |^2} )$.
Hence, we
introduce an auxiliary variable $g$ to  maximize  $\min (  {| \bm{h}_1^{H}(\boldPhi) {\bm{w}} |^2},\ldots, {| \bm{h}_i^{H}(\boldPhi) {\bm{w}} |^2},\ldots, {| \bm{h}_K^{H}(\boldPhi) {\bm{w}} |^2} )$ where ${| \bm{h}_i^{H}(\boldPhi) {\bm{w}} |^2}=\tr(\bm{A}_i\boldsymbol{V} )  + {b}_{i} {b}_{i}^H$, which is equivalent to  maximizing  $\min (  {| \bm{h}_1^{H}(\boldPhi) \overline{\bm{w}} |^2},\ldots, {| \bm{h}_K^{H}(\boldPhi) \overline{\bm{w}} |^2} )$,
 because ${|\bm{h}^H_{i}(\boldPhi)\bm{w}  |^2}$ $={P_t| \bm{h}_i^{H}(\boldPhi) \overline{\bm{w}} |}^2$ and $P_t$ is  constant during this step. 
 Thus,  Problem~(P6) can be further transformed to
\begin{subequations} \label{eq-Prob-P1e}
\begin{alignat}{2}
\setlength{\abovedisplayskip}{3pt plus 1pt minus 1pt}
\setlength{\belowdisplayskip}{3pt plus 1pt minus 1pt}
\setlength\abovedisplayshortskip{3pt plus 1pt minus 1pt}
\setlength\belowdisplayshortskip{3pt plus 1pt minus 1pt}
\text{(P9)}:  \max_{\boldsymbol{V},g} g & \\
\mathrm{s.t.}~~~~~&   \tr(\bm{A}_i\boldsymbol{V} )  + {b}_{i} {b}_{i}^H \geq g + \gamma \sigma^2, ~ i =1,\ldots,K, \label{eq:prob-P1e-prime-SINR-constraint}\\
 &\boldsymbol{V} \succeq 0,~ \ g \geq 0, \boldsymbol{V}_{n,n}  = 1,~ n=1,\ldots,N+1.
\end{alignat}
\end{subequations}
We can utilize the CVX software~\cite{CVX} to solve Problem (P9).  Generally, the SDR may not produce a rank-one solution to Problem (P9). Thus, once $\bm{V}$ is available, we use again the Gaussian randomization~\cite{Gaussianrandom} method to obtain multiple candidate rank-one solutions  to  Problem (P9), and  we select the one with the maximum value of $\min (  {| \bm{h}_1^{H}(\boldPhi) {\bm{w}} |^2},\ldots, {| \bm{h}_K^{H}(\boldPhi) {\bm{w}} |^2} )$ as the value of $\boldPhi$ during the $q^{th}$ iteration, which is denoted by  $\boldPhi^{{(q)}}$.

{\color{black}\begin{prop}
The proposed SCA-based  alternating algorithm is convergent.
\end{prop}
\begin{proof}
Similarly, the convergence of the proposed SCA-based  alternating algorithm is guaranteed by the
following two facts. First, the objective value in Problem (P2)
 is non-increasing over iterations. More specifically, during each iteration, the variable $g$ satisfies $g\ge 0$,
  which means  $ f^{{(q)}}_{\text{o}\boldPhi}\ge f^{(q)}_{\text{o}\bm{w}}$. Furthermore, if  $\bm{w}^{(q+1)}$ is the optimal solution to Problem $\text{(P}8)$ during the $(q+1)^{th}$ iteration, we obtain $f^{(q+1)}_{\text{o}\bm{w}} \ge f^{{(q)}}_{\text{o}\boldPhi}$.  Therefore, we have $f^{(q+1)}_{\text{o}\bm{w}} \ge f^{{(q)}}_{\text{o}\boldPhi}\ge f^{(q)}_{\text{o}\bm{w}}$, indicating that $P_t^{(q+1)} \le P^{(q)}_t$. Hence, the transmit power at the BS $P_t$ is \mbox{non-increasing} over the iterations.
 Second, the optimal
value  is bounded from below due to the SNR
constraints for Problem (P2). Therefore, the proposed SCA-based  alternating algorithm is 
convergent.
\end{proof}}

\subsection{Complexity analysis} \label{sec1:Optimization2}

{\color{red} The  procedure  of  the  SDR-based  algorithm  is  shown  in  Algorithm  1, which shows  that steps 4 and 9 take a major part of the complexity because steps 4 and 9 solve Problems (P3) and (P7), respectively.
The complexities of solving Problems (P3) and (P7) include the complexities of 1) solving the SDP problem by utilizing the CVX tool and 2) Gaussian randomization process.  CVX first transforms Problems (P3) and (P7) into the standard SDP problems and then utilizes  an  interior  point  method (IPM) to obtain the solution. There are  $(K+M^2)$ variances in the standard  SDP form of Problem (P3), and the IPM requires $\mathcal{O}\left( \sqrt{K+M^2}\right)$ iterations and costs $\mathcal{O}\left( {(K+M^2)}^{3}\right)$ arithmetic  operations in each iteration to obtain a  solution~\cite{complexity}. Furthermore, the eigenvalue decomposition of $\bm{X}$ of size $M\times M$ takes a major part of the computational complexity in the  Gaussian randomization process, which requires $\mathcal{O}(M^3)$  arithmetic operations~\cite{2-5}. Then, the complexity of  solving Problem (P3) is $\mathcal{O}\left( {(K+M^2)}^{3.5}+M^3\right)$.
Similarly,  the complexity of  solving Problem (P7) is $\mathcal{O}\left( {(2K+(N+1)^2)}^{3.5}+(N+1)^3\right)$.
Finally, the complexity of the  SDR-based  algorithm is $\mathcal{O}\left(L_\text{SDR} \left({(K+M^2)}^{3.5}+M^3+{(2K+(N+1)^2)}^{3.5}+(N+1)^3\right)\right)$
where  $L_\text{SDR}$   denotes the number of iterations of the SDR-based algorithm.

The procedure  of  the  SCA-based  algorithm  is  shown  in  Algorithm  2,  and  steps 4 and 9 take a major part of  complexity because steps 4 and 9 solve Problems (P8) and (P9), respectively.
Similar to  Problem (P3),  
the complexity of  solving Problem (P9) is $\mathcal{O}\left(( K+(N+1)^2)^{3.5}+(N+1)^3\right)$ for one  iteration. Furthermore, the SCA method obtains a  solution to 
Problem (P8) by solving  ${\text{(P}}{{\text{8}}^{{'}}}{\text{)}}$ iteratively, and we use $I_\text{SCA}^\text{P8}$ to denote  the number of iterations. Similar to Problem (P3), we  compute the complexity of the standard form of  ${\text{(P}}{{\text{8}}^{{'}}}{\text{)}}$ which is a second-order cone programming (SOCP). 
Let $q$ be the number of Linear Matrix Inequality (LMI) constraints,  $m$ be the number of SOC constraints,  $k_j$ denote the size of the $j^{th}$ LMI or SOC constraint, and $n$ represent the total number of variables of the standard  form of  ${\text{(P}}{{\text{8}}^{{'}}}{\text{)}}$. Then, 
the complexity of solving the standard  ${\text{(P}}{{\text{8}}^{{'}}}{\text{)}}$  is $\sqrt{\alpha}C$ where $\sqrt{\alpha}=\sqrt{\sum\nolimits_{j=1}^{q} k_j+2m}$ is the iteration complexity and  $C= n\sum\nolimits_{j=1}^{q} k_j^{3}+n^2\sum\nolimits_{j=1}^{q} k_j^{2}+n\sum\nolimits_{j=1}^{m} k_j^{2}+n^3$ is the per-iteration computation cost~\cite{2-6}. 
 Note that, $\{x_i,y_i\}$ in Problem ${\text{(P}}{{\text{8}}^{{{'}}}}{\text{)}}$ can be expressed as expressions in CVX and thus we do not regard them as variances.
For the standard  form of Problem ${\text{(P}}{{\text{8}}^{{'}}}{\text{)}}$, the total number of variables is $n=M+1$ and there are $K$ LMIs of size 1 and one SOC of size $M+1$ in constraints. Therefore, the total complexity of the SCA-based method is $\mathcal{O}\left(I_\text{SCA}\left(I_\text{SCA}^\text{P8}\sqrt{K+2}(M+1)^{3}+( K+(N+1)^2)^{3.5}+(N+1)^3 \right)\right)$, where $I_\text{SCA}$ is the number of alternating optimization iterations for the SCA-based method.

We can see that 
the highest complexity orders  of the SCA-based  algorithm  and the SDR-based algorithm are ($M^{3}$, $N^{7}$)  and  ($M^{7}$, $N^{7}$). So, the
 complexity of the SDR-based  algorithm is
much higher than that of the SCA-based  algorithm 
for the  massive MIMO scenarios.  Furthermore, the simulation results in Section VI will show that the SCA-based method  outperforms the SDR-based method. Note that,  all computations are executed at the BS with strong data computing capability and do not increase the amount of calculation on the MEs and RIS sides. 
}

 \section{Lower bounds for the  average transmit power} \label{sec:analysis}
 In this section,  we  first derive an
 analytical lower  bound and then a tighter semi-analytical  lower  bound.
In particular, we show how the lower bounds depend on the number of RIS units $N$,  the  number of MEs $K$, and  the number of antennas $M$. The following two case studies are considered: \mbox{1) $K=1$ and $M>1$}, and \mbox{2) $K>1 $ and $M>1$.} In addition, when   discussing   the setup $K=1$, we omit  the subscript $i$ (ME index) of  $\beta_{\text{b},i}$ and $\beta_{\text{r},i}$ for  ease of presentation. For benchmarking purposes, we also derive analytical lower bounds  for the transmit power without the RIS and  with random phase shifts at the RIS.

  We consider the independent and identically distributed (i.i.d)  Rayleigh fading channel  model for the RIS-$i^{th}$ME and BS-$i^{th}$ME links{\color{red}\footnote{{\color{red}
 The BS and the RIS are fixed in position while MEs may be in motion. Commonly, the direct line of BS-ME or RIS-ME  is  obstructed by buildings or something else. In such case, the amplitude
fluctuations of the received signal follow Rayleigh distribution~\cite{1-6}, thereby adopting the Rayleigh fading channels for the BS-ME and RIS-ME links.
In particular,  when analyzing  transmission
technologies,  the typical way is 
to consider the tractable i.i.d Rayleigh fading channel model~\cite{1-8}.
~In fact, the   i.i.d Rayleigh fading channel is  reasonable  in  isotropic scatterer environment  when the RIS and the BS are in uniform linear array (ULA)  with $\lambda/2$-spacing,  as in~\cite{1-8,1-10,1-11}.
}}}, i.e.~$\bm{h}_{\text{r},i}\sim \text{ }\mathcal{C}\mathcal{N}(0,{\beta_{\text{r},i}^2}{\bm{I}})$ and~$\bm{h}_{\text{b},i}\sim \text{ }\mathcal{C}\mathcal{N}(0,{\beta_{\text{b},i}^2}{\bm{I}})$ where $\beta_{\text{r},i}^2$ and $\beta_{\text{b},i}^2$ account for  the path loss, and a full-rank LoS channel model for the BS-RIS link. Let ($x_{\text{BS}}$, $y_{\text{BS}}$, $z_{\text{BS}}$) and ($x_{\text{RIS}}$, $y_{\text{RIS}}$, $z_{\text{RIS}}$) be the  coordinates of the BS  and the RIS, respectively.
Let $\bm{H}_{\text{b},\text{r},m,n}$
denote  the channel  response between  the $m^{th}$ antenna of the BS  and
the $n^{th}$ element at the RIS.  Then,   the full-rank LoS channel  between the BS and the RIS is given by~\cite{Los}
\begin{equation}\label{eqH}
\setlength{\abovedisplayskip}{2pt plus 1pt minus 1pt}
\setlength{\belowdisplayskip}{2pt plus 1pt minus 1pt}
\setlength\abovedisplayshortskip{2pt plus 1pt minus 1pt}
\setlength\belowdisplayshortskip{2pt plus 1pt minus 1pt}
\begin{aligned}
\bm{H}_{\text{b},\text{r},m,n}&={\beta_{\text{b,r}}}\hspace{-1.5pt}\exp \left( {j\frac{{2\pi }}{\lambda }{d_{{\rm{BS}}}}}({(m - 1)\text{sin}{\phi _{Lo{S_1}(n)}}\text{sin}{\theta _{Lo{S_1}(n)}}})\right) \\
&\quad \times \exp \left( {j\frac{{2\pi }}{\lambda }{d_{{\rm{RIS}}}}}((n - 1)\text{sin}{\phi _{Lo{S_2}(n)}}\text{sin}{\theta _{Lo{S_2}(n)}}) \right),~m =1,\ldots,M,~n =1,\ldots,N,
\end{aligned}
\end{equation}
 where $\lambda $ is the wavelength, $d_\text{BS}$ and $d_\text{RIS}$ are the inter-antenna distances at the BS and the RIS, respectively.
  ${\phi _{Lo{S_1}}}(n)$ and ${\phi _{Lo{S_2}}}(n)$ denote the azimuth angles at the BS and the RIS, respectively,
  ${\theta _{Lo{S_1}}}(n)$ and ${\theta _{Lo{S_2}}}(n)$ denote the  elevation angle of departure at the BS and the elevation angle of arrival at the  RIS, respectively, and $\beta_{\text{b,r}}$ accounts for  the path loss of the BS-RIS channel.
   ${\phi _{Lo{S_1}}}(n)$  and ${\theta _{Lo{S_1}}}(n)$ are generated uniformly between 0 to  $2\pi$ and 0 to  $\pi$, respectively, and satisfy ${\phi _{Lo{S_2}}}(n) = \pi  + {\phi _{Lo{S_1}}}(n), {\theta _{Lo{S_2}}}(n) = \pi  - {\theta _{Lo{S_1}}}(n)$~\cite{Los}. We can see from ({\ref{eqH}}) that, the path losses between  any antenna at the BS and any  element at the RIS are the same (i.e., $\beta_{\text{b,r}}^2$ in ({\ref{eqH}})). This is because the distance  between the BS and  the  RIS is relatively large compared to the size  of  the  RIS~{\cite{RISimpleDi}},  such as the far-field regime~{\cite{TBA}}. 


\subsection{Lower bounds for the average transmit power in  RIS aided systems}

Since the  BS-MEs and RIS-MEs channels are  random variables,  we focus  on analyzing the lower bound for the average transmit power in  RIS aided systems{\color{red}\footnote{\color{red}The  proposed  optimization  algorithms  are  not  for  some  particular  channels.
 The proposed method of deriving the lower bounds for  the  RIS-aided  system is based on  the LoS channel and  i.i.d  Rayleigh fading channel, which requires the first and second  central moments of the amplitudes of considered channels.
The  spatially correlated Rician fading  channel  couples  the amplitudes of the LoS channel and  spatially correlated Rayleigh fading channel~\cite{IRSwu2}. Based on the triangle inequality, the spatially correlated Rician fading channel  can be decoupled into the LoS channel and spatially correlated Rayleigh fading channel in terms of  amplitude. Furthermore, the spatially correlated Rayleigh fading  channel  can  be  represented  by a linear combination of  multiple  i.i.d Rayleigh fading channels~\cite{1-16}. Thus,   based on the triangle inequality,  the spatially correlated Rayleigh fading  channel can be  decoupled into the sum of multiple  i.i.d  Rayleigh fading channels in terms of  amplitude. 
Therefore, the method of deriving lower bounds in our paper can be easily extended to the spatially correlated Rician fading channel.

}}.  We first present an analytical lower bound. To get a tighter lower bound, we further derive a semi-analytical lower bound.


\subsubsection{Analytical lower bound }\label{sec:closed-bound}
\

\textbf{Case 1): $K=1$ and $M>1$}


To obtain the minimum transmit power, the QoS constraint inequality in Problem (P1) needs to be fulfilled with equality, i.e., ${{P_t| \bm{h}_1^{H}(\boldPhi) \overline{\bm{w}} |}^2}= \gamma{{\sigma}^2 }$. In fact, the optimized  $\boldPhi$ and $\bm{w}$  maximize the value of $| \bm{h}_1^{H}(\boldPhi)\overline{\bm{w}}|^2$ since $\gamma{{\sigma}^2 }$ is constant, thus obtaining the  minimum transmit power ${P_t}$.

 The BS-MEs and RIS-MEs channels are  random variables. For each realization of these random variables,  we can obtain the optimized $\boldPhi$ and $\bm{w}$ by utilizing our proposed optimization methods, and can obtain the maximum value of $| \bm{h}_1^{H}(\boldPhi)  \overline{\bm{w}} |^2$.  Then,  we can  obtain the  average   maximum value of $| \bm{h}_1^{H}(\boldPhi)  \overline{\bm{w}} |^2$, and thus
the average  transmit power can be formulated as
\begin{equation}\label{L-hp}
\setlength{\abovedisplayskip}{2pt plus 1pt minus 1pt}
\setlength{\belowdisplayskip}{2pt plus 1pt minus 1pt}
\setlength\abovedisplayshortskip{2pt plus 1pt minus 1pt}
\setlength\belowdisplayshortskip{2pt plus 1pt minus 1pt}
\overline{P_t}=\frac{{\sigma}^2\gamma}{{\mathbb{E}(\max(| \bm{h}_1^{H}(\boldPhi)  \overline{\bm{w}} |^2))}}.
\end{equation}

As for $|\bm{h_1}^{H}(\boldPhi) \overline{\bm{w}}|$, we have~\cite{wu2018intelligentfull}
\begin{equation}\label{L-h3}
 \setlength{\abovedisplayskip}{2pt plus 1pt minus 1pt}
\setlength{\belowdisplayskip}{2pt plus 1pt minus 1pt}
\setlength\abovedisplayshortskip{1pt plus 1pt minus 1pt}
\setlength\belowdisplayshortskip{1pt plus 1pt minus 1pt}
\begin{aligned}
|  \bm{h}_1^{H}(\boldPhi) \overline{\bm{w}}  | &=| \bm{h}_{\text{r},1}^{{H}}{\boldPhi}{\bm{H}_{\text{b},\text{r}}} \overline{\bm{w}}+  \bm{h}_{\text{b},1}^{{H}} \overline{\bm{w}}|
\mathop  \le \limits^{(a)}  |\bm{h}_{\text{r},1}^{{H}}{\boldPhi}{\bm{H}_{\text{b},\text{r}}} \overline{\bm{w}}|+|\bm{h}_{\text{b},1}^{{H}} \overline{\bm{w}}|.
\end{aligned}
\end{equation}
Based on the triangle inequality, inequality (a) takes the equality sign  if and only if
$\arg(\bm{h}_{\text{r},1}^{{H}}{\boldPhi}{\bm{H}_{\text{b},\text{r}}} \overline{\bm{w}})=\arg(\bm{h}_{\text{b},1}^{{H}} \overline{\bm{w}})=\varphi_0$~\cite{wu2018intelligentfull}.
Define ${A}:=|\bm{h_{\text{r},1}}^{{H}}{\boldPhi}{\bm{H}_{\text{b},\text{r}}} \overline{\bm{w}}|, {B}:=|h_{\text{b},1}^{{H}} \overline{\bm{w}}|$. Then, based on (\ref{L-h3}), we have~\cite{han2019intelligent}
\begin{equation}\label{L-h4}
\setlength{\abovedisplayskip}{2pt plus 1pt minus 1pt}
\setlength{\belowdisplayskip}{2pt plus 1pt minus 1pt}
\setlength\abovedisplayshortskip{1pt plus 1pt minus 1pt}
\setlength\belowdisplayshortskip{1pt plus 1pt minus 1pt}
\begin{aligned}
\mathbb{E}(\max(|\bm{h}_1^{H}(\boldPhi)  \overline{\bm{w}}|^2))
=\mathbb{E}(({A}+{B})^2)
=\mathbb{E}({A}^2)+2\mathbb{E}({AB})+\mathbb{E}({B}^2).
\end{aligned}
\end{equation}
Thus, to derive the lower bound for the  average  transmit power $\overline{P_t}$, we need to compute the  maximum value of $\mathbb{E}(\max(|\bm{h}_1^{H}(\boldPhi)  \overline{\bm{w}}|^2))$ with respect to ${\boldPhi}$~and ${\overline{
\bm{w}}}$~\cite{han2019intelligent}, denoted  by  $Q_1$, i.e.,
\begin{equation}\label{L-2h5a}
\setlength{\abovedisplayskip}{2pt plus 1pt minus 1pt}
\setlength{\belowdisplayskip}{2pt plus 1pt minus 1pt}
\setlength\abovedisplayshortskip{1pt plus 1pt minus 1pt}
\setlength\belowdisplayshortskip{1pt plus 1pt minus 1pt}
\begin{aligned}
 Q_1& =\max_{\boldPhi,\overline{
\bm{w}}}(\mathbb{E}({A}^2)+2\mathbb{E}({AB})+\mathbb{E}({B}^2)).
\end{aligned}
\end{equation}
Next, we discuss  how to compute  $\mathbb{E}({A}^2)$, $\mathbb{E}({AB})$ and $\mathbb{E}({B}^2)$, thereby deriving  $ Q_1$.

Define $ C_n:={{\sum\limits_{m=1}^{M}}{\bm{H}_{\text{b},\text{r},m,n}}  \overline{w_m}}$ \ \ $(n=1,\ldots,N)$. Then,  we   have
\begin{equation}\label{Cn}
\setlength{\abovedisplayskip}{2pt plus 1pt minus 1pt}
\setlength{\belowdisplayskip}{2pt plus 1pt minus 1pt}
\setlength\abovedisplayshortskip{1pt plus 1pt minus 1pt}
\setlength\belowdisplayshortskip{1pt plus 1pt minus 1pt}
\begin{aligned}
{\rm{|}}{C_n}{{\rm{|}}^2} &= {\left| {\sum\limits_{m = 1}^M {{H_{\text{b},\text{r},m,n}}\overline{w_m}} } \right|^2} = {\left| {{H_{\text{b},\text{r},1,n}}\overline{w_1} + \cdots + {H_{\text{b},\text{r},M,n}}\overline{w_M}} \right|^2}\\
&\mathop  \le \limits^{(a)} \sum\limits_{m = 1}^M {{{\left| {{H_{\text{b},\text{r},m,n}}\overline{w_m}} \right|}^2}}  + 2\sum\limits_{t = 1}^M {\sum\limits_{k = t + 1}^M {\left| {{H_{\text{b},\text{r},t,n}}\overline{w_t}} \right|} } \left| {{H_{\text{b},\text{r},k,n}}\overline{w_k}} \right|\\
&= {\left| {{H_{\text{b},\text{r},1,n}}} \right|^2} + \left( {2{{\left| {{H_{\text{b},\text{r},1,n}}} \right|}^2}\sum\limits_{t = 1}^M {\sum\limits_{k = t + 1}^M {\left| {\overline{w_t}} \right|} } \left| {\overline{w_k}} \right|} \right)
\mathop  \le \limits^{(b)} {\left| {{H_{\text{b},\text{r},1,n}}} \right|^2}M \mathop  = \limits^{(c)}  {M\beta _{\text{b,r}}^2},
\end{aligned}
\end{equation}
where step~(a)~follows from the fact that ${\left| {\sum\nolimits_{m = 1}^M {{H_{\text{b},\text{r},m,n}}\overline{w_m}} } \right|} \le  \sum\nolimits_{m = 1}^M {{{\left| {{H_{\text{b},\text{r},m,n}}\overline{w_m}} \right|}}}$, step~(b)~follows  from  the fact  that the term $(|\overline {{w_k}} ||\overline {{w_t}} |)$ takes  the maximum value if $|\overline {{w_k}} |=|\overline {{w_t}} |$, which yields $|\overline {{w_1}} |^2=\cdots=|\overline {{w_M}}|^2=\frac{1}{M}$ because  $\sum\nolimits_{m=1}^{M}| \overline{w_m}|^2=1$, and step~(c)~follows from the fact that the elements in ${\bm{H}_{\text{b},\text{r}}}$ have the same amplitude  since  the distance  between the BS and  the  RIS is relatively large compared to the size  of  the  RIS.

As for $\mathbb{E}({A}^2)$, we have
\begin{subequations}\label{L-2h1}
\setlength{\abovedisplayskip}{2pt plus 1pt minus 1pt}
\setlength{\belowdisplayskip}{2pt plus 1pt minus 1pt}
\setlength\abovedisplayshortskip{1pt plus 1pt minus 1pt}
\setlength\belowdisplayshortskip{1pt plus 1pt minus 1pt}
\begin{align}
\mathbb{E}(A)&\mathop =\limits^{(a)}\mathbb{E}\left( {\sum\limits_{n=1}^{N}}{|h_{\text{r},1,n}^{H}}|\left| {{\sum\limits_{m=1}^{M}}{{H}_{\text{b},\text{r},m,n}}  \overline{w_m} }\right|  \right)
\mathop =\limits^{(b)}\mathbb{E}({|h_{\text{r},1,1}^{H}}|)\left({\sum_{n=1}^{N}} |C_n|\right)\mathop \le \limits^{(c)} \frac{\sqrt{\pi M} N{\beta_{\text{b,r}}} \beta_r}{2}, \\
\mathbb{E}({A}^2)
    &=\mathbb{E}\left(\left( {\sum\limits_{n=1}^{N}}{|h_{\text{r},i,n}^{H}}| |C_n| \right)^2\right) \nonumber\\
    &=\mathbb{E}\left( {\sum\limits_{n=1}^{N}}\left({|h_{\text{r},1,n}^{H}}|^2 |C_n|^2 \right)+ 2\sum\limits_{n=1}^{N}\sum\limits_{i=n+1}^{N}|h_{\text{r},1,n}^{H}||C_n||h_{\text{r},1,i}^{H}||C_i| \right) \nonumber \\
    &=\mathbb{E}\left(|h_{\text{r},1,1}^{H}|^2\right){\sum\limits_{n=1}^{N}}\left(|C_n|^2\right) + 2\mathbb{E}^2\left(|h_{\text{r},1,1}^{H}|\right)\sum\limits_{n=1}^{N}\sum\limits_{i=n+1}^{N}|C_n||C_i| \nonumber\\
   & \mathop \le \limits^{(d)}
    \frac{\pi N^2{\beta_{\text{b,r}}}^2 \beta_r^2M}{4}\hspace{-1pt}+\hspace{-1pt}\frac{N\beta_r^2{\beta_{\text{b,r}}}^2M}{2}(2\hspace{-1pt}-\hspace{-1pt}\frac{\pi}{2}), \label{A-equation-1}
\end{align}
\end{subequations}
where
step~(a)~follows from the fact that $\arg(\bm{h_{\text{r},1}}^{{H}}{\boldPhi}{\bm{H}_{\text{b},\text{r}}} \overline{\bm{w}})=\varphi_0$, step~(b)~follows from the fact that $\bm{h_{\text{r},1}}\sim \text{ }\mathcal{C}\mathcal{N}(0,{\beta_{\text{r},1}^2}{\bm{I}})$, step~(c)~is derived based on the Cauchy--Schwartz inequality,
and steps (c) and (d)  follow from  the fact  that  $|C_n|^2\le {{M}{\beta_{\text{b,r}}}^2}~(n=1,\ldots,N)$ as shown in (\ref{Cn}) and that $| {h_{\text{r},1,1}}^{H}|$  has a Rayleigh distribution with variance $\frac{\beta_r^2}{2}(2-\frac{\pi}{2})$.

As for $\mathbb{E}({B}^2)$, we have~\cite{han2019intelligent}
\begin{subequations}\label{L-2h2}
\setlength{\abovedisplayskip}{2pt plus 1pt minus 1pt}
\setlength{\belowdisplayskip}{2pt plus 1pt minus 1pt}
\setlength\abovedisplayshortskip{2pt plus 1pt minus 1pt}
\setlength\belowdisplayshortskip{2pt plus 1pt minus 1pt}
\begin{alignat}{2}
\mathbb{E}({B})&=\mathbb{E}\left( {\left|\sum\limits_{m=1}^{M} {h_{\text{b},m}}^{H} \overline{w_m}\right|}\right)\mathop \le\limits^{(a)}\mathbb{E}\left( \sum\limits_{m=1}^{M}| {h_{\text{b},m}}^{H} \overline{w_m}|\right)
  =\mathbb{E}(|h_{\text{b},1}^{{H}}|)\sum\limits_{m=1}^{M}| \overline{w_m}|,\\
\mathbb{E}^2(B)  &= \mathbb{E}^2(|h_{\text{b},1}^{{H}}|)\left(\sum\limits_{m=1}^{M}| \overline{w_m}|\right)^2
\mathop  \le \limits^{(b)} (M|\overline{w_1}|)^2\frac{\beta_b^2 \pi}{4}=\frac{\pi \beta_b^2 M}{4},\\
\mathbb{E}({B}^2)
    &\le \mathbb{E}\left(\left(\sum\limits_{m=1}^{M}| {h_{\text{b},m}}^{H} | | \overline{w_m}| \right)^2\right) \nonumber\\
    &=\mathbb{E}\left( {\sum\limits_{m=1}^{M}}\left({|h_{\text{b},m}^{H}}|^2 | \overline{w_m}|^2 \right)+ 2\sum\limits_{m=1}^{M}\sum\limits_{i=m+1}^{M}|h_{\text{b},m}^{H}||\overline{w_m}||h_{\text{b},m}^{H}||\overline{w_i}| \right) \nonumber \\
    &=\mathbb{E}\left(|h_{\text{b},1}^{H}|^2\right) {\sum\limits_{m=1}^{M}}\left(|\overline{w_m}|^2\right)+ 2\mathbb{E}^2|h_{\text{b},1}^{H}|\sum\limits_{m=1}^{M}\sum\limits_{i=M+1}^{M}\|\overline{w_m}||\overline{w_i}| 
   \mathop  \le \limits^{(b)} \frac{\pi \beta_b^2 M}{4}+ \frac{\beta_b^2}{2}(2-\frac{\pi}{2}),\label{B-equation-1}
\end{alignat}
\end{subequations}
where step (a)~follows from the fact that $ {\left|\sum\nolimits_{m=1}^{M} {h_{\text{b},m}}^{H} \overline{w_m}\right|}\le \sum\nolimits_{m=1}^{M}| {h_{\text{b},m}}^{H} \overline{w_m}|$, steps (b) and (c)~follow from the fact that $| {h_{\text{b},1}}^{H}|$  has a Rayleigh distribution with mean  $\frac{ \beta_\text{b} {\sqrt{\pi }}  }{2}$ and variance $\frac{\beta_\text{b}^2}{2} (2-\frac{\pi}{2})$, step (b)  follows  from  the fact  that the  term $(| \overline{w_1}|+| \overline{w_2}|+\cdots+| \overline{w_M}|)^2$ takes  the maximum value if $| \overline{w_1}|=| \overline{w_2}|=\cdots=| \overline{w_M}|$ because $\sum\nolimits_{m=1}^{M}| \overline{w_m}|
^2=1$.

As for $\mathbb{E}(AB)$,  given $\overline{\bm{w}}$, variables $A$ and $B$ are independent of each other. Hence, we have
\begin{equation}\label{L-2h3}
\setlength{\abovedisplayskip}{2pt plus 1pt minus 1pt}
\setlength{\belowdisplayskip}{2pt plus 1pt minus 1pt}
\setlength\abovedisplayshortskip{3pt plus 1pt minus 1pt}
\setlength\belowdisplayshortskip{3pt plus 1pt minus 1pt}
\begin{aligned}
&\mathbb{E}(AB)=\sqrt{\mathbb{E}^2(A)\mathbb{E}^2(B)}\le \frac{N\pi \beta_r {\beta_{\text{b,r}}} \beta_b  M}{4}. \\
\end{aligned}
\end{equation}
When  $|\overline{w_1}|=\cdots=|\overline{w_M}|$  and $|C_1|=\cdots=|C_N|$, (\ref{A-equation-1}), (\ref{B-equation-1}) and (\ref{L-2h3}) hold with the equality sign. 
Then, we  can formulate $Q_1$ as follows
\begin{equation}\label{L-2h5}
\begin{aligned}
 Q_1& =\max_{\boldPhi,\overline{
\bm{w}}}(\mathbb{E}({A}^2)+2\mathbb{E}({AB})+\mathbb{E}({B}^2))\\
   &=\frac{\pi N^2{\beta_{\text{b,r}}}^2 \beta_r^2M}{4}+\frac{N\beta_r^2{\beta_{\text{b,r}}}^2M}{2}(2-\frac{\pi}{2}) +\frac{N\pi \beta_r {\beta_{\text{b,r}}} \beta_b   M}{2} +\frac{\beta_b^2}{2}(2-\frac{\pi}{2})+\frac{\pi \beta_b^2 M}{4}.
\end{aligned}
\end{equation}

Based on (\ref{L-hp}),
a lower bound for the average  transmit power at the BS is given by
\begin{equation}\label{L-011}
\begin{aligned}
&  P_{K=1,M>1}^L=\dfrac{{\sigma}^2\gamma}{Q_1}=\dfrac{{\sigma}^2\gamma}{\resizebox{0.7\hsize}{!}{$\dfrac{\pi N^2\beta_b^2 \beta_r^2M}{4}+\dfrac{N\beta_r^2{\beta_{\text{b,r}}}^2M}{2}(2-\dfrac{\pi}{2}) +\dfrac{N\pi \beta_r {\beta_{\text{b,r}}} \beta_b   M}{2}+\dfrac{\beta_b^2}{2}(2-\dfrac{\pi}{2})+\dfrac{\pi \beta_b^2 M}{4}$}},
\end{aligned}
\end{equation}
where the superscript ``$L$" is used to denote ``lower bound".
(\ref{L-011}) confirms that the average transmit power of an RIS scales with $1/{N^2}$~\cite{Nnew}.

\textbf{Case 2): $K>1 $ and $M>1$}

Based on Problem (P1) and (\ref{min-power}), the average  transmit power  $\overline{P_t}$  is~\cite{han2019intelligent}
\begin{equation}\label{L-3h2}
\overline{P_t}= \frac{\gamma\sigma^2}{\min ( {\mathbb{E}( \max({| \bm{h_1}^{H}(\boldPhi)  \overline{\bm{w}} |^2}))},\ldots,  {\mathbb{E}(\max({| \bm{h_K}^{H}(\boldPhi)  \overline{\bm{w}} |^2}) )})}.
\end{equation}
By using the inequality ${\min ( {\mathbb{E}( \max({| \bm{h_1}^{H}(\boldPhi)  \overline{\bm{w}} |^2}))},\ldots,\mathbb{E}(\max({| \bm{h_i}^{H}(\boldPhi)  \overline{\bm{w}} |^2})),\ldots,{\mathbb{E}(\max({| \bm{h_K}^{H}(\boldPhi)  \overline{\bm{w}} |^2} ))})}$
$\le {\min({Q_1,\ldots,Q_i,\ldots,Q_K})}$, a lower bound for the average transmit power at the BS is~\cite{han2019intelligent}
\begin{equation}\label{L-002}
P_{K>1,M>1}^L= \frac{\gamma\sigma^2}{\min({Q_1,\ldots,Q_i,\ldots,Q_K})},
\end{equation}
where  $Q_i$ is the maximum value of $\mathbb{E}(\max(| {h_i}^{H}(\boldPhi) \overline{\bm{w}} |^2))$ with respect to ${\boldPhi}$~and ${\overline{
\bm{w}}}$. (\ref{L-2h1})--(\ref{L-2h5}) allow us to compute $Q_1$.  A similar approach is used to compute other values of  $Q_i$ ($i=1,\ldots,K$).

\subsubsection{Semi-analytical  lower bound}

\

\textbf{Case 1): $K=1$ and $M>1$}

 The maximum value of  $ \mathbb{E}(\max(|\bm{h}_1^{H}(\boldPhi)  \overline{\bm{w}}|^2))$ with respect to $\boldPhi$ and $\overline{\bm{w}}$  is  \\$ Q_1 =\mathop {\max }\limits_{\boldPhi,\overline{
\bm{w}}}(\mathbb{E}({A}^2)+2\mathbb{E}({AB})+\mathbb{E}({B}^2))$ in (\ref{L-2h5a}). In the following, we compute each term of  $Q_1$. Define $ C_n:={{\sum\nolimits_{m=1}^{M}}{\bm{H}_{\text{b},\text{r},m,n}}  \overline{w_m}}$. Then, we  have
\begin{equation}\label{L-2c1}
  \begin{aligned}
|{C_n}{|} &=\sqrt{ {\left| {\sum\limits_{m = 1}^M {{H_{\text{b},\text{r},m,n}}} \overline {{w_m}} } \right|^2}{\rm{   }}}\vspace{10pt}
 = \sqrt{{\left| {\sum\limits_{m = 1}^M {|{H_{\text{b},\text{r},m,n}}} ||\overline {{w_m}} |\text{exp}(j({\varphi _{{H_{\text{b},\text{r},m,n}}}} + {\varphi _{{w_m}}}))} \right|^2}{\rm{  }}}\\
 &= \sqrt{{\left| {{H_{\text{b},\text{r},1,1}}} \right|^2}{\left| {\sum\limits_{m = 1}^M | \overline {{w_m}} |\text{exp}(j({\varphi _{{H_{\text{b},\text{r},m,n}}}} + {\varphi _{{w_m}}}))} \right|^2}}
  \mathop =\limits^{(a)}  \sqrt{{\left| {{H_{\text{b},\text{r},1,1}}} \right|^2}{\left| {\sum\limits_{m = 1}^M | \overline {{w_m}} |\text{exp}(j({\varphi _{_{m,n}}}))} \right|^2}}\vspace{10pt} \\
  &=\sqrt{{\left| {{H_{\text{b},\text{r},1,1}}} \right|^2}\left| {\sum\limits_{m = 1}^M | \overline {{w_m}} |(\text{cos}({\varphi _{_{m,n}}}) + j\text{sin}({\varphi _{_{m,n}}})}) \right|^2}\\
   & = \sqrt{{\left| {{H_{\text{b},\text{r},1,1}}} \right|^2}\left( {\sum\limits_{m = 1}^M | \overline {{w_m}} {|^2}({\cos^2}({\varphi _{_{m,n}}}) + {{\sin }^2}({\varphi _{_{m,n}}}))}
+2\sum\limits_{k = 1}^M {\sum\limits_{t = k + 1}^M {(|\overline {{w_k}} ||\overline {{w_t}} |(\cos({\varphi _{_{k,n}}} - {\varphi _{_{t,n}}})))} }  \right)}\vspace{10pt} \\
& \mathop \le\limits^{(b)}\sqrt { {{{{({{\beta_{\text{b,r}}}})}^2}}}\left( {1 + \frac{2}{M}\sum\limits_{k = 1}^M {\sum\limits_{t = k + 1}^M {(({\cos}({\varphi _{_{k,n}}} - {\varphi _{_{t,n}}})))} } } \right)} \vspace{10pt}\\
 &\mathop =\limits^{(c)}\sqrt{ {{{{{({\beta _{\text{b,r}}})}^2}}}\left( {1 + \frac{2}{M}\sum\limits_{k = 1}^M {\sum\limits_{t = k + 1}^M {(({\cos}(\text{const}{_{k,t,n}} + {\varphi _{{w_k} - {w_t}}})))} } } \right)} }
 \mathop =\limits^{(d)}\sqrt{ {{{{{({\beta _{\text{b,r}}})}^2}}}\left( {1 + \frac{2}{M}\sum\limits_{k = 1}^M {\sum\limits_{t = k + 1}^M {f_\text{c}({{{w_k}}, {{w_t}}})} } } \right)} },
\end{aligned}
\end{equation}
where   step (a) is derived by defining$ {{\varphi _{_{m,n}}} := {\varphi _{{H_{\text{b},\text{r},m,n}}}} + {\varphi _{{w_m}}}}$ where $\varphi_c$ denotes the angle of a complex number $c$, step (b)  follows  from  the fact  that the term $(|\overline {{w_k}} ||\overline {{w_t}} |)$
takes  the maximum value if $|\overline {{w_k}} |=|\overline {{w_t}} |$, which implies $|\overline {{w_1}}|^2=\cdots=|\overline {{w_M}} |^2=\frac{1}{M}$ because  $\sum\limits_{m=1}^{M}| \overline{w_m}|^2=1$,
step (c) is obtained by using ${\varphi _{_{k,n}}} - {\varphi _{_{t,n}}} = {\varphi _{{H_{\text{b},\text{r},k,n}}}} - {\varphi _{{H_{\text{b},\text{r},t,n}}}} + {\varphi _{{w_k}}} - {\varphi _{{w_t}}}= \text{const}{_{k,t,n}} + {\varphi _{{w_k}}} - {\varphi _{{w_t}}}$ and  defining ${\varphi _{{w_k}}} - {\varphi _{{w_t}}}:={\varphi _{{w_k}-{w_t}}}$, and  step (d) is derived by defining  $ f_\text{c}({{{w_k}}, {{w_t}}}):=\text{cos}(\text{const}{_{k,n}} + {\varphi _{{w_k} - {w_t}}})$.

As for $\mathbb{E}({A}^2)$, we have
\begin{subequations}\label{L-2h11}
\begin{align}
\mathbb{E}(A)&\mathop =\limits^{(a)}\mathbb{E}\left( {\sum_{n=1}^{N}}{|h_{\text{r},1,n}^{H}}|\left| {{\sum_{m=1}^{M}}{{H}_{\text{b},\text{r},m,n}}  \overline{w_m} }\right|\right) 
\mathop=\limits^{(b)}\frac{\beta_r \sqrt{ \pi}}{2}\left({\sum_{n=1}^{N}} |C_n|\right),\\
\mathbb{E}({A}^2)
    &=\mathbb{E}\left(\left( {\sum\limits_{n=1}^{N}}{|h_{\text{r},i,n}^{H}}| |C_n| \right)^2\right) \nonumber\\
    &=\mathbb{E}\left( {\sum\limits_{n=1}^{N}}\left({|h_{\text{r},i,n}^{H}}|^2 |C_n|^2 \right)+ 2\sum\limits_{n=1}^{N}\sum\limits_{i=n+1}^{N}|h_{\text{r},i,n}^{H}||C_n||h_{\text{r},1,i}^{H}||C_i| \right) \nonumber \\
    &=\mathbb{E}\left(|h_{\text{r},1,1}^{H}|^2\right){\sum\limits_{n=1}^{N}}\left(|C_n|^2\right) + 2\mathbb{E}^2\left(|h_{\text{r},1,1}^{H}|\right)\sum\limits_{n=1}^{N}\sum\limits_{i=n+1}^{N}|C_n||C_i| \nonumber\\
    &\mathop =\limits^{(d)}{\sum\limits_{n=1}^{N}}\left(\beta_r^2 \left(|C_n|^2\right) \right)+ 2\sum\limits_{n=1}^{N}\sum\limits_{i=n+1}^{N}\frac{ (\beta_\text{r})^2 {{\pi }}  }{4}|C_n||C_i|,
\end{align}
\end{subequations}
where step (a)~follows from  $\arg(\bm{h_{\text{r},1}}^{{H}}{\boldPhi}{\bm{H}_{\text{b},\text{r}}} \overline{\bm{w}})=\varphi_0$, 
 and steps~(b) and (c)~follow from the fact that $| {h_{\text{r},1,1}}^{H}|$  has a  Rayleigh distribution with mean  ${ \beta_r {\sqrt{\pi }}  }/{2}$ and variance ${\beta_r^2}/{2} (2-{\pi}/{2})$.

Similarly, as for $\mathbb{E}({B}^2)$, we have
\begin{equation}
\setlength{\abovedisplayskip}{2pt plus 1pt minus 1pt}
\setlength{\belowdisplayskip}{2pt plus 1pt minus 1pt}
\setlength\abovedisplayshortskip{3pt plus 1pt minus 1pt}
\setlength\belowdisplayshortskip{3pt plus 1pt minus 1pt}
\begin{aligned}\nonumber
\mathbb{E}({B^2}) &= \mathbb{E}\left( {{{\left| {\sum\limits_{m = 1}^M {{h_{\text{b},1,m}}} \overline {{w_m}} } \right|}^2}} \right){\rm{   }} 
 = \mathbb{E}\left( {{{\left| {\sum\limits_{m = 1}^M {|{h_{\text{b},1,m}}} ||\overline {{w_m}} |\exp(j({\varphi _{{h_{\text{b},1,m}}}} + {\varphi _{{w_m}}}))} \right|}^2}{\rm{ }}} \right){\rm{ }}\vspace{10pt}\\
 &=\mathbb{E}\left( {{{\left| {\sum\limits_{m = 1}^M | {{h_{\text{b},1,m}}}|| \overline {{w_m}} |\cos({{\varphi _{{h_{\text{b},1,m}}}}+\varphi _{w_m}}) + j\sum\limits_{m = 1}^M | {{h_{\text{b},1,m}}}|| \overline {{w_m}} |\sin({\varphi _{{h_{\text{b},1,m}}}}+{\varphi_{w_m}})} \right|}^2}} \right)\vspace{10pt} \\
&=\mathbb{E}\left( {{{ {\left(\sum\limits_{m = 1}^M | {{h_{\text{b},1,m}}}|| \overline {{w_m}} |\cos({\varphi _{{h_{\text{b},1,m}}}}+{\varphi _{w_m}})\right)^2 + \left(\sum\limits_{m = 1}^M | {{h_{\text{b},1,m}}}|| \overline {{w_m}} |\sin({\varphi _{{h_{\text{b},1,m}}}}+{\varphi _{w_m}})\right)^2 }}}} \right)\vspace{10pt} \\
&=  \mathbb{E}\left( {\sum\limits_{m = 1}^M | {{h_{\text{b},1,m}}}|^2| \overline {{w_m}} {|^2}(\cos{^2}({\varphi _{{h_{\text{b},1,m}}}}+{\varphi _{w_m}}) + {{\sin }^2}{(\varphi _{{h_{\text{b},1,m}}}}+{\varphi _{w_m}}))} \right)\vspace{10pt} \\
&\quad+ \mathbb{E}\left( { 2\sum\limits_{k = 1}^M {\sum\limits_{t = k + 1}^M {(| {{h_{\text{b},1,k}}}||\overline {{w_k}} || {{h_{\text{b},1,t}}}||\overline {{w_t}} |(\cos({\varphi _{{h_{\text{b},1,k}}}}-{\varphi _{{h_{\text{b},1,t}}}}+{\varphi _{w_k}} - {\varphi _{w_t}})))} } } \right)\vspace{10pt} \\
&\mathop =\limits^{(a)}\left(\frac{{{{({\beta _\text{b}})}^2}\pi }}{4}+ \frac{\beta_\text{b}^2}{2}(2-\frac{\pi}{2})\right)\left( {\sum\limits_{m = 1}^M | \overline {{w_m}} {|^2}} \right)+ \frac{{{{({\beta _\text{b}})}^2}\pi }}{4}\left( { 2\sum\limits_{k = 1}^M {\sum\limits_{t = k + 1}^M {(|\overline {{w_k}} ||\overline {{w_t}} |((\cos({\varphi_{h,k,t}+\varphi_{{w_k} - {w_t}}})))} } } \right)
\end{aligned}
\end{equation}

\begin{equation}\label{L-2h21}
\setlength{\abovedisplayskip}{2pt plus 1pt minus 1pt}
\setlength{\belowdisplayskip}{2pt plus 1pt minus 1pt}
\setlength\abovedisplayshortskip{3pt plus 1pt minus 1pt}
\setlength\belowdisplayshortskip{3pt plus 1pt minus 1pt}
\begin{aligned}
&\mathop \le \limits^{(b)} \frac{{{{({\beta _b})}^2}\pi }}{4}\left( {1 + \frac{2}{M}\sum\limits_{k = 1}^M {\sum\limits_{t = k + 1}^M {((\cos({\varphi_{h,k,t}+\varphi_{{w_k} - {w_t}}})))} } } \right)+\frac{\beta_b^2}{2}(2-\frac{\pi}{2}) \\
&\mathop =\limits^{(c)}\frac{{{{({\beta _b})}^2}\pi }}{4}\left( {1 + \frac{2}{M}\sum\limits_{k = 1}^M {\sum\limits_{t = k + 1}^M {{f_{\rm{b}}}({{{w_k}},  {{w_t}}})} } } \right)+\frac{\beta_\text{b}^2}{2}(2-\frac{\pi}{2}),
\end{aligned}
\end{equation}
where step (a)~is derived by   the fact that $| {h_{\text{b},1,1}}^{H}|$  has a  Rayleigh distribution with mean  $\frac{ \beta_b {\sqrt{\pi }}  }{2}$ and defining ${\varphi _{{h_{\text{b},1,k}}}}-{\varphi _{{h_{\text{b},1,t}}}}+{\varphi _{w_k}} - {\varphi _{w_t}}:={\varphi_{h,k,t}+\varphi_{{w_k} - {w_t}}}$, step (b)  follows  from  the fact  that the term $(|\overline {{w_k}} ||\overline {{w_t}} |)$
takes  the maximum value if $|\overline {{w_k}} |=|\overline {{w_t}} |$ which implies $|\overline {{w_1}} |^2=\ldots=|\overline {{w_M}} |^2=\frac{1}{M}$ because  $\sum\limits_{m=1}^{M}| \overline{w_m}|^2=1$, 
and step (c) is derived by defining  $ f_\text{b}({{{w_k}}, {{w_t}}}):=\cos({\varphi_{h,k,t}+\varphi_{{w_k} - {w_t}}})$.

As for $\mathbb{E}(AB)$, since $\mathbb{E}^2(B) \le \frac{\pi \beta_b^2 M}{4}$  based on (\ref{L-2h2}),  we have
\begin{equation}\label{L-ab}
\setlength{\abovedisplayskip}{2pt plus 1pt minus 1pt}
\setlength{\belowdisplayskip}{2pt plus 1pt minus 1pt}
\setlength\abovedisplayshortskip{3pt plus 1pt minus 1pt}
\setlength\belowdisplayshortskip{3pt plus 1pt minus 1pt}
\begin{aligned}
&\mathbb{E}(AB)=\sqrt{\mathbb{E}^2(A)\mathbb{E}^2(B)}\le\left({\sum_{n=1}^{N}} {\sqrt {{{{{({\beta _{\text{b,r}}})}^2}}}\left( {1 + \frac{2}{M}\sum\limits_{k = 1}^M {\sum\limits_{t = k + 1}^M {f_\text{c}({{{w_k}}, {{w_t}}})} } } \right)} }\right) \frac{\pi\beta_r{ \beta_b \sqrt{M}}}{4}.
\end{aligned}
\end{equation}
Based on (\ref{L-2c1})-(\ref{L-ab}), we evince that the upper bounds of $\mathbb{E}({A}^2)$, $\mathbb{E}({B}^2)$ and $\mathbb{E}(AB)$ are related to variables ${{w_1}},\ldots,{ {w_M}}$.  Therefore, the upper bound of $(\mathbb{E}({A}^2)+\mathbb{E}({B})^2+2\mathbb{E}({AB}))$ is related to  variables ${{w_1}},\ldots,{ {w_M}}$, which is denoted by $f_1({{w_1}},\ldots,{ {w_M}})$ and  can be written as follows
\begin{equation}\label{L-2h51}
\setlength{\abovedisplayskip}{2pt plus 1pt minus 1pt}
\setlength{\belowdisplayskip}{2pt plus 1pt minus 1pt}
\setlength\abovedisplayshortskip{3pt plus 1pt minus 1pt}
\setlength\belowdisplayshortskip{3pt plus 1pt minus 1pt}
\begin{aligned}
  & \mathbb{E}({A}^2)+\mathbb{E}({B})^2+2\mathbb{E}({AB})\le f_1({{w_1}},\ldots,{ {w_M}})    \\
  &=\left({\sum_{n=1}^{N}} {\sqrt {{{({\beta _{\text{b,r}}})}^2}\left( {1 + \frac{2}{M}\sum\limits_{k = 1}^M {\sum\limits_{t = k + 1}^M {f_\text{c}({{{w_k}}, {{w_t}}})} } } \right)} }\right)^2 \frac{\beta_r^2 \pi}{4}   \\
   &\quad +\frac{\beta_r^2}{2}(2\hspace{-1pt}-\hspace{-1pt}\frac{\pi}{2})\hspace{-1pt}\times\hspace{-1pt} \left({\sum_{n=1}^{N}} {{{{({\beta _{\text{b,r}}})}^2}\left( {1 + \frac{2}{M}\sum\limits_{k = 1}^M {\sum\limits_{t = k + 1}^M {f_\text{c}({{{w_k}}, {{w_t}}})} } } \right)} }\right)
   +\frac{{{{({\beta _b})}^2}\pi }}{4}\left( {1 + \frac{2}{M}\sum\limits_{k = 1}^M {\sum\limits_{t = k + 1}^M {{f_{\rm{b}}}({{{w_k}}, {{w_t}}})} } } \right)\\
   &\quad+\frac{\beta_b^2}{2}(2-\frac{\pi}{2})+2\left({\sum_{n=1}^{N}} {\sqrt {{{({\beta _{\text{b,r}}})}^2}\left( {1 + \frac{2}{M}\sum\limits_{k = 1}^M {\sum\limits_{t = k + 1}^M {f_\text{c}({{{w_k}}, {{w_t}}})} } } \right)} }\right) \frac{\pi\beta_r{ \beta_b \sqrt{M}}}{4}.
\end{aligned}
\end{equation}

Thus, ${Q_1}$ is the maximum value of $f_1({{w_1}},\ldots,{ {w_M}})$, which can be formulated as follows
\begin{equation}\label{L-Q11}
\setlength{\abovedisplayskip}{2pt plus 1pt minus 1pt}
\setlength{\belowdisplayskip}{2pt plus 1pt minus 1pt}
\setlength\abovedisplayshortskip{3pt plus 1pt minus 1pt}
\setlength\belowdisplayshortskip{3pt plus 1pt minus 1pt}
{Q_1} = \mathop {\max }\limits_{{{w_1}},\ldots,{ {w_M}}} f_1({{w_1}},\ldots,{ {w_M}}).
\end{equation}

Based on (\ref{L-hp}),
a lower bound for the average  transmit power at the BS is given by
\begin{equation}\label{L-0111}
\setlength{\abovedisplayskip}{2pt plus 1pt minus 1pt}
\setlength{\belowdisplayskip}{2pt plus 1pt minus 1pt}
\setlength\abovedisplayshortskip{3pt plus 1pt minus 1pt}
\setlength\belowdisplayshortskip{3pt plus 1pt minus 1pt}
\begin{aligned}
P_{K=1,M>1}^L&=\dfrac{{\sigma}^2\gamma}{Q_1}
=\dfrac{{\sigma}^2\gamma}{ \mathop {\max }\limits_{{{w_1}},\ldots,{ {w_M}}} f_1({{w_1}},\ldots,{ {w_M}})}.
\end{aligned}
\end{equation}


  Since the value of $f_1({{w_1}},\ldots,{ {w_M}})$ depends on  $ f_\text{c}({{{w_k}}, {{w_t}}}):=((\text{cos}(\text{const}{_{k,n}} + {\varphi _{{w_k} - {w_t}}})))$ and $ f_\text{b}({{{w_k}}, {{w_t}}}):=\text{cos}( {\varphi _{{w_k} - {w_t}}})$  based on (\ref{L-2h51}), $f_1$ depends on the phase difference of two antennas  at the BS (i.e., ${\varphi _{{w_k} - {w_t}}}$) which ranges from 0 to $2\pi$. Therefore, we  utilize a brute-force method to find  phase differences that maximize the value of $f_1$. This means that a lower bound for the average transmit power for such case  is in a semi-analytical form.

\textbf{Case 2): $K>1 $ and $M>1$}

By using the inequality ${\min ( {\mathbb{E}( \max({| \bm{h_1}^{H}(\boldPhi)  \overline{\bm{w}} |^2}))},\ldots,\mathbb{E}(\max({| \bm{h_i}^{H}(\boldPhi)  \overline{\bm{w}} |^2})),\ldots,  {\mathbb{E}(\max({| \bm{h_K}^{H}(\boldPhi)  \overline{\bm{w}} |^2} ))})}$ \\
$\le {\min({Q_1,\ldots,Q_i,\ldots,Q_K})}$ and  (\ref{L-Q11}),  a lower bound for the average  transmit power at the BS is given by
\begin{equation}\label{L-0021}
\setlength{\abovedisplayskip}{2pt plus 1pt minus 1pt}
\setlength{\belowdisplayskip}{2pt plus 1pt minus 1pt}
\setlength\abovedisplayshortskip{3pt plus 1pt minus 1pt}
\setlength\belowdisplayshortskip{3pt plus 1pt minus 1pt}
\begin{aligned}
 P_{K>1,M>1}^L&= \frac{\gamma\sigma^2}{\min({Q_1,Q_2,\ldots,Q_K})}
=\frac{\gamma\sigma^2}{\mathop {\max }\limits_{{{w_1}},\ldots,{ {w_M}}}\min({{ f_1,f_2,\ldots,f_K}})}.
\end{aligned}
\end{equation}
 (\ref{L-2c1})--(\ref{L-Q11}) allow us to compute $f_1$.  A similar approach can be used to compute other values of  $f_i$ ($i=1,\ldots,K$).
 Based on (\ref{L-Q11}), we can derive that ({\ref{L-0021}})  depends on the phase difference of two antennas  at the BS  which  ranges from 0 to $2\pi$. We can also utilize a brute-force method  to find  phase differences that maximize the value of
  $\min({{ f_1,f_2,\ldots,f_K}})$.

\begin{rem}[{\it{Relation between the analytical lower bound and the
semi-analytical lower  bound}}]
{\rm{
{\color{black}When all of the $\cos$ terms in step (c) of (31) and step (b) of (33) equal the maximum value $1$, terms $|{C_n}{|}$ and  $\mathbb{E}(B^2)$ in the  \mbox{semi-analytical} lower bound  are equal to those in  the  \mbox{analytical} lower bound, respectively. Then, it is easy to derive that the \mbox{ semi-analytical} lower bound  is equal to the  \mbox{analytical} lower bound. In turn, if all of the $\cos$ terms in step (c) of (\ref{L-2c1}) and step (b) of (\ref{L-2h21}) equal the maximum value $1$, we have ${\varphi_{h,k,t}+\varphi _{{w_k}}} - {\varphi _{{w_t}}}=0$ and  $\text{const}{_{k,n}} + {\varphi _{{w_k}}} - {\varphi _{{w_t}}}=0$, where  $\varphi_c$ denotes the angle of a complex number $c$. We derive that $ \text{const}{_{k,t,n}}=\varphi_{h,k,t}~    \text{i.e.},~ {\varphi _{{H_{\text{b},\text{r},k,n}}}} - {\varphi _{{H_{\text{b},\text{r},t,n}}}}={\varphi _{{h_{\text{b},1,k}}}}-{\varphi _{{h_{\text{b},1,t}}}}$ where $\bm{H}_{\text{b},\text{r},k({{or}}~ t),n}$ denotes  the channel  response between  the $(k(or~ t))^{th}$ antenna of the BS  and the $n^{th}$ element at the RIS. 
However, based on the full-rank LoS channel model in (\ref{eqH}) and the  uncorrelated Rayleigh fading channel between the BS and each ME, this does not hold. Hence, it is impossible to ensure that all of the $\cos$ terms in step (c) of (\ref{L-2c1}) and step (b) of (\ref{L-2h21}) equal the maximum value $1$, which indicates that the semi-analytical bound can never equal the analytical bound. Compared to the \mbox{analytical} lower bound, the  semi-analytical lower bound is closer to the simulation results.  However, the \mbox{analytical} lower bound can intuitively show  advantages of the RIS.
}}}
\end{rem}

\subsection{Analytical  lower bound for the average  transmit power with random phase shifts at the RIS}

\textbf{Case 1):
$K=1$ and $M>1$}

Similar to  RIS aided systems,  the maximum value of  $ \mathbb{E}(\max(|\bm{h}_1^{H}(\boldPhi)  \overline{\bm{w}}|^2))$ with respect to  $\overline{\bm{w}}$  is~$ Q_1 ={\mathop {\max }\limits_{\overline{
\bm{w}}}}(\mathbb{E}({A}^2)+2\mathbb{E}({AB})+\mathbb{E}({B}^2))$. We discuss  how to compute  each term of  $Q_1$.

Recalling ${A}=|\bm{h_{\text{r},1}}^{{H}}{\boldPhi}{\bm{H}_{\text{b},\text{r}}} \overline{\bm{w}}|={\left|{\sum\nolimits_{n=1}^{N}}{h_{r,1,n}^{H}}e^{j \theta_1^{(n)}}{{\sum\nolimits_{m=1}^{M}}{{H}_{\text{b},\text{r},m,n}}  \overline{w_m} }\right|}={\left|{\sum\nolimits_{n=1}^{N}}{h_{\text{r},1,n}^{H}}e^{j \theta_1^{(n)}} C_n\right|} $, we obtain that
 ${\sum\nolimits_{n=1}^{N}}{h_{\text{r},1,n}^{H}}e^{j \theta_1^{(n)}} C_n$ has a circularly-symmetric complex Gaussian distribution with mean $0$  and variance $(|C_1|^2+\cdots+|C_N|^2)\textup{Var}(h_{r,1,1})=(|C_1|^2+\cdots+|C_N|^2)\beta_r^2$. Then, $A$ follows a Rayleigh distribution with mean $\frac{ \sqrt{|C_1|^2+\cdots+|C_N|^2}\beta_r {\sqrt{\pi }}  }{2}$ and variance $\frac{(|C_1|^2+\cdots+|C_N|^2)\beta_r^2}{2} (2-\frac{\pi}{2})$. Hence,  we have
\begin{equation}\label{L-2h13}
\setlength{\abovedisplayskip}{2pt plus 1pt minus 1pt}
\setlength{\belowdisplayskip}{2pt plus 1pt minus 1pt}
\setlength\abovedisplayshortskip{3pt plus 1pt minus 1pt}
\setlength\belowdisplayshortskip{3pt plus 1pt minus 1pt}
\begin{aligned}
\mathbb{E}(A^2)&=\mathbb{E}^2({A})+\textup{Var}(A)
               =(\frac{ \sqrt{|C_1|^2+\cdots+|C_N|^2}\beta_r {\sqrt{\pi }}  }{2})^{2}+\frac{(|C_1|^2+\cdots+|C_N|^2)\beta_r^2}{2} (2-\frac{\pi}{2})\\
               &=(|C_1|^2+\cdots+|C_N|^2)\beta_r^2\left(\frac{\pi}{4} + \frac{1}{2}(2-\frac{\pi}{2})\right)
               \mathop  \le \limits^{(a)}{{NM}{\beta_{\text{b,r}}}^2}\beta_r^2,
\end{aligned}
\end{equation}
where step  (a) derives  from ({\ref{Cn}}), i.e.,  $|C_n|^2\le {{M}{\beta_{\text{b,r}}}^2} \ \ (n=1,\ldots,N)$.

 $\mathbb{E}({B}^2)$ is the same as we derived in ({\ref{L-2h2}}), i.e., $\mathbb{E}(B^2)
\leq \frac{\pi \beta_b^2 M}{4}+ \frac{\beta_b^2}{2}(2-\frac{\pi}{2})$.

As for $\mathbb{E}(AB)$, since $\mathbb{E}^2(B) \le \frac{\pi \beta_b^2 M}{4}$  based on (\ref{L-2h2}),  we have
\begin{equation}\label{L-2h33}
\setlength{\abovedisplayskip}{2pt plus 1pt minus 1pt}
\setlength{\belowdisplayskip}{2pt plus 1pt minus 1pt}
\setlength\abovedisplayshortskip{3pt plus 1pt minus 1pt}
\setlength\belowdisplayshortskip{3pt plus 1pt minus 1pt}
\begin{aligned}
&\mathbb{E}(AB)=\sqrt{\mathbb{E}^2(A)\mathbb{E}^2(B)} \le \frac{\sqrt{N}\pi \beta_r {\beta_{\text{b,r}}} \beta_b  M}{4}. \\
\end{aligned}
\end{equation}

When  $|\overline{w_1}|=\cdots=|\overline{w_M}|$  and $|C_1|=\cdots=|C_N|$, (\ref{L-2h13}), ({\ref{L-2h2}}), and (\ref{L-2h33})  hold with  the equality  sign. Then, $Q_1$ can be formulated as  follows
\begin{equation}\label{L-2h53}
\setlength{\abovedisplayskip}{2pt plus 1pt minus 1pt}
\setlength{\belowdisplayskip}{2pt plus 1pt minus 1pt}
\setlength\abovedisplayshortskip{3pt plus 1pt minus 1pt}
\setlength\belowdisplayshortskip{3pt plus 1pt minus 1pt}
\begin{aligned}
   Q_1  &=\max_{\overline{
\bm{w}}}(\mathbb{E}({A}^2)+2\mathbb{E}({AB})+\mathbb{E}({B}^2))
   ={{NM}{\beta_{\text{b,r}}}^2}\beta_r^2 +\frac{\sqrt{N}\pi \beta_r {\beta_{\text{b,r}}} \beta_b  M}{2}+\frac{\pi \beta_b^2 M}{4}+ \frac{\beta_b^2}{2}(2-\frac{\pi}{2}).
\end{aligned}
\end{equation}
Based on  (\ref{L-hp}), a lower bound for the average  transmit power at the BS is given by
\begin{equation}\label{L-0113}
\setlength{\abovedisplayskip}{2pt plus 1pt minus 1pt}
\setlength{\belowdisplayskip}{2pt plus 1pt minus 1pt}
\setlength\abovedisplayshortskip{3pt plus 1pt minus 1pt}
\setlength\belowdisplayshortskip{3pt plus 1pt minus 1pt}
\begin{aligned}
 P_{K=1,M>1}^L&=\dfrac{{\sigma}^2\gamma}{Q_1}
=\dfrac{{\sigma}^2\gamma}{{{NM}{\beta_{\text{b,r}}}^2}\beta_r^2 +\frac{\sqrt{N}\pi \beta_r {\beta_{\text{b,r}}} \beta_b  M}{2}+\frac{\pi \beta_b^2 M}{4}+ \frac{\beta_b^2}{2}(2-\frac{\pi}{2})}.
\end{aligned}
\end{equation}
(\ref{L-0113}) confirms that the average transmit power  with random phase shifts at the RIS nearly scales with
 $1/{N}$.

\textbf{Case 2): $K>1 $ and $M>1$}

The lower bound expression for such case is the same as ({\ref{L-002}}), where the value of  $Q_i$ ($i=1,\ldots,K$) can be computed  by utilizing the same method used for obtaining $Q_1$  in (\ref{L-2h53}).

\subsection{Analytical lower bound for the average transmit power  without the  RIS}

\textbf{Case 1): $K=1$ and $M>1$}

For such case, there exist only the direct channels between the BS and MEs. Hence, we have $Q_1={\mathop {\max }\limits_{\overline{
\bm{w}}}} ~\mathbb{E}({B}^2)$. Based on ({\ref{L-2h2}}), we have $Q_1={\mathop {\max }\limits_{\overline{
\bm{w}}}} ~\mathbb{E}(B^2)
= \frac{\pi \beta_b^2 M}{4}+ \frac{\beta_b^2}{2}(2-\frac{\pi}{2})$.
 Then,  a lower bound for the average  transmit power at the BS  is
\begin{equation}\label{L-0114}
\setlength{\abovedisplayskip}{2pt plus 1pt minus 1pt}
\setlength{\belowdisplayskip}{2pt plus 1pt minus 1pt}
\setlength\abovedisplayshortskip{3pt plus 1pt minus 1pt}
\setlength\belowdisplayshortskip{3pt plus 1pt minus 1pt}
\begin{aligned}
P_{K=1,M>1}^L&=\dfrac{{\sigma}^2\gamma}{Q_1}
=\dfrac{{\sigma}^2\gamma}{\frac{\pi \beta_b^2 M}{4}+ \frac{\beta_b^2}{2}(2-\frac{\pi}{2})}.
\end{aligned}
\end{equation}
\textbf{Case 2): $K>1 $ and $M>1$}

The lower bound expression for such case is the same as ({\ref{L-002}}), where the value of  $Q_i$ ($i=1,\ldots,K$) can be computed  by utilizing the same method used for obtaining $Q_1$  in ({\ref{L-2h2}}).

 
{\color{red} \begin{rem}[{\it{What happens if the  CSI is in error}}]
{\rm{ 
For Problem (P1), the
SNR constraint (i.e.,  the 
QoS  requirement) is represented by the  CSI.
If the estimated CSI is in error, the  CSI can be represented by  the sum of the estimated CSI and the estimation error. However, we can  only obtain  the  estimated CSI  and  cannot  know the estimation  error in practice.
Therefore, to formulate the problem under imperfect CSI estimation, 
the key is to  tackle the CSI estimation error in SNR and  ensure the QoS requirement at the same time. Then, the  optimization algorithms can be designed to solve the optimization problem under imperfect CSI estimation.
Authors in~\cite{2-7} have dealt with how to tackle such problem, but they only considered the imperfect RIS-ME CSI estimation and assumed perfect BS-RIS and BS-ME CSI estimation. The constraints under  imperfect RIS-ME, BS-RIS, and RIS-ME CSI estimations are more complex than the case in~\cite{2-7}, which are left for our future work.
The  performance  of  the  proposed optimization algorithms and the derived lower bounds in this paper can be viewed as  benchmarks.
}}
\end{rem}}

{\color{red}\begin{rem}[{\it{Differences between our paper, \cite{R1} and the single-ME optimization problem in~\cite{wu2018intelligentfull}}}]
{\rm{ Both our paper, \cite{R1}, and~\cite{wu2018intelligentfull} aim to minimize the transmit power at the BS.  The main differences between our paper and~\cite{R1}~\cite{wu2018intelligentfull} are discussed as follows.


Both \cite{R1} and our paper aim to minimize the  transmit power at the BS over broadcasting signals 1) but for different RIS-aided systems. More specifically,~\cite{R1} was  for symbiotic radio  system to support passive Internet of Things (IoT), where the RIS is not only to reflect signals from the BS but also needs to transmit data to other MEs. In contrast, our paper is to support active IoT where the RIS is only used  to reflect signals from the BS.
~2) Authors in~\cite{R1} decoupled the original problem into two subproblems and  utilized the SDR method to solve them, which is similar to our proposed SDR-based method but our approach is more general. The reason is that it is the action of the RIS transmitting  data to other MEs  in~\cite{R1} that  makes the constraints for  decoding the transmitted signal at the BS in~\cite{R1} be  a particular case of the constraints in our problem. 3) To reduce the computational complexity and improve the performance of the SDR-based optimization algorithm, we  also  propose an SCA-based alternating  optimization  algorithm.
4) Though authors in \cite{R1} solved the problem of  discrete phase shifts, the discrete phase shifts were  obtained by just quantifying the derived continuous phase shifts.
 Therefore, the key is to design  the optimization algorithms for the continuous phase shifts problem in~\cite{R1}. Our paper considers the continuous phase shifts. 
 We can  obtain the discrete  phase  shifts by quantifying the derived continuous phase  shifts or other optimization algorithms, which are left for our future work. The  performance of the proposed optimization algorithms with continuous phase shifts in this paper can be viewed as a benchmark.
Furthermore, we  derive two lower bounds of the average transmit power at the BS  in our paper further to analyze  the effectiveness of our proposed optimization algorithms. In contrast, authors in~\cite{R1} did not present  the  lower bounds to  analyze the effectiveness of their proposed optimization algorithm. The method of deriving the lower bounds can be utilized for the problem in~\cite{R1}. The reason is that the lower bounds for the  problem in~\cite{R1} and our problem are determined by the constraints, and the constraints in~\cite{R1} are the special case of the constraints in our paper.

For \cite{wu2018intelligentfull}  and our paper,  the RIS is only used to  reflect  signals  from  the  BS  to  MEs. 1) From the perspective of  the problem formulation, the  only  difference  is considering  the  SNR  constraints  at  multiple  MEs in our paper instead  of the SNR  constraint at  single  ME in the single-ME optimization problem of~\cite{wu2018intelligentfull}. However, this difference induces  significant differences  in problem-solving between our paper and~\cite{wu2018intelligentfull}. 
Specifically, for the single-ME optimization problem in~\cite{wu2018intelligentfull}, the maximum-ratio transmission is the
optimal transmit beamforming solution to this single-ME optimization problem  when given phase shifts,   and the phase shifts at the RIS can be tuned to achieve the maximum
combined channel  gains for  a signal ME   at
the same time. Based on these observations,  authors in ~\cite{wu2018intelligentfull} proposed two approaches to solve the single-ME optimization problem, which cannot be utilized  to solve our problem for multiple MEs. The reason is that, for multiple MEs, our problem  is NP-hard  when  given phase shifts~\cite{complexity}, and the phase shifts cannot be tuned to achieve the maximum
combined channel  gains for  all MEs   at
the same time because all MEs with  different channels share the same phase shifts. The single-ME optimization problem in~\cite{wu2018intelligentfull} is a particular case of our problem (i.e., $K=1$), and thus our proposed optimization algorithms can solve the single-ME optimization problem in~\cite{wu2018intelligentfull}   as well, which means that our methods are more general.
2) Furthermore,   we  derive analytical  and semi-analytical  lower  bounds  for the  average  transmit  power at the BS for the general case (i.e.,  multiple  MEs  and a  multi-antenna  BS) to analyze the effectiveness of the proposed optimization algorithms.
Authors in~\cite{wu2018intelligentfull} just presented the analytical result of scaling law of the transmit power at the BS with  the number of elements at the RIS 
 based  on  a particular case  (i.e.,  a  single  ME  and a  single-antenna  BS),  aiming to obtain  an  insight into the advantage of the RIS. 3) In addition, we present the detailed procedure of  how to derive lower bounds based on the LoS BS-to-RIS channel and i.i.d Rayleigh fading BS-to-ME and RIS-to-ME channels and briefly show  how to extend our derived lower bounds to the generalized spatially correlated Rician fading channels. The analytical result in~\cite{wu2018intelligentfull}  is only based on i.i.d Rayleigh fading  BS-RIS and RIS-ME channel and ignoring the BS-ME channel. Thus,  the method of deriving lower bounds in our paper is more general. 

}}

\end{rem}}


%
%

\section{Simulation results} \label{sec:simulation}

In  this section, we  utilize  numerical results to validate the   proposed optimization algorithms  and the    derived lower bounds. We assume that  a BS with a uniform linear array of antennas is located at $(0,0,0)$,  and  an RIS with a uniform linear array of RIS units is located at $(0,50,0)$. The inter-antenna and inter-unit separation at the BS and the RIS are a half of the wavelength. The purpose of deploying the RIS is to improve the signal strength.  To illustrate this benefit, we assume that the MEs are uniformly located  within the half-circle centered at the RIS with radius 3 m as shown in Fig.~\ref{simulation-lacation}, which are the cell-edge MEs.
\begin{figure}[h]
\setlength{\belowcaptionskip}{-0.5cm}
\setlength{\abovecaptionskip}{-0.1cm}
  \centering
\includegraphics[scale=0.8]{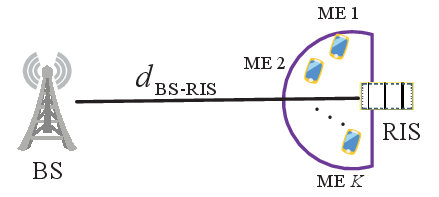}
  \caption{ The location of the RIS, BS and MEs in the simulation.
  }\label{simulation-lacation}
\end{figure}

 The channel models for the BS-RIS, BS-MEs and RIS-MEs are the same as those  we described in Section \ref{sec:analysis}. The path loss between $a$ and $b$ is $C_0(d_{a,b}/D_0)^{-\alpha}$, where~$C_0=1$~m, $D_0=-30$~dB, $d_{a,b}$ denotes the distance between $a$ and $b$, and $\alpha$ is the path loss exponent. We set $\sigma^2=$ -30 dBm,   and $\varepsilon=10^{-4}$. For the BS-RIS, RIS-MEs, and BS-MEs, we set $\alpha$  to be $2,2.8$, and $3.5$, respectively. As the baselines, we employ the conventional power control (i.e., MMSE and ZF based beamforming, 
 termed ``Without-RIS-MMSE" and ``Without-RIS-ZF", respectively, in the  figures),  power control with random phase shifts at the RIS  (termed
``Random-RIS" in the  figures), and {\color{black} a two-stage algorithm proposed in~\cite{wu2018intelligentfull}}. The MMSE-based beamforming is obtained by solving Problem (P2) and setting $\boldPhi=\bm{0}$, and the ZF-based beamforming is obtained according to ${\rm{trace}}\left( {{\rm{diag}}\left( {{{({\sigma _1})}^2}{\gamma _1},\ldots,{{({\sigma _K})}^2}{\gamma _K}} \right) \times {{\left( {\bm{H_b}^H\bm{H_b}} \right)}^{ - 1}}} \right)$~\cite{wu2018intelligentfull}. In addition, the terms ``alternating optimization algorithm based on  SDR",  ``alternating optimization algorithm based on SCA", ``analytical lower bound for the  average transmit power in  RIS aided systems", and ``semi-analytical lower bound for the average transmit power in RIS aided systems"  are abbreviated as  ``With-RIS-SDR", ``With-RIS-SCA", ``With-RIS-LB1" and ``With-RIS-LB2" in the  figures, respectively. Terms  `` analytical lower bound for the average transmit power with random phase shifts at the RIS"  and ``analytical lower bound for the transmit power without  the RIS" are abbreviated as  ``Random-RIS-LB"  and ``Without-RIS-LB" in the figures, respectively.

Note that, for the single-ME case, since the result obtained by the MMSE-based beamforming is the same as that obtained by the ZF-based beamforming, we use the term ``Without-RIS"  to denote the conventional power control scheme. Furthermore, the semi-analytical lower bound for the average transmit power depends on the phase difference between any two antennas. We utilize the brute-force method to obtain the semi-analytical lower bound with 500 quantization levels when  $M$ is small. To resolve the search explosion when  $M$ is sufficiently large, we utilize a random search method~\cite{Randomsearch}~to obtain the semi-analytical lower bound.
\begin{figure}[h]
\setlength{\belowcaptionskip}{-0.5cm}
\setlength{\abovecaptionskip}{-0.2cm}
\centering
\begin{minipage}{.5\textwidth}
  \centering
  \includegraphics[scale=0.3]{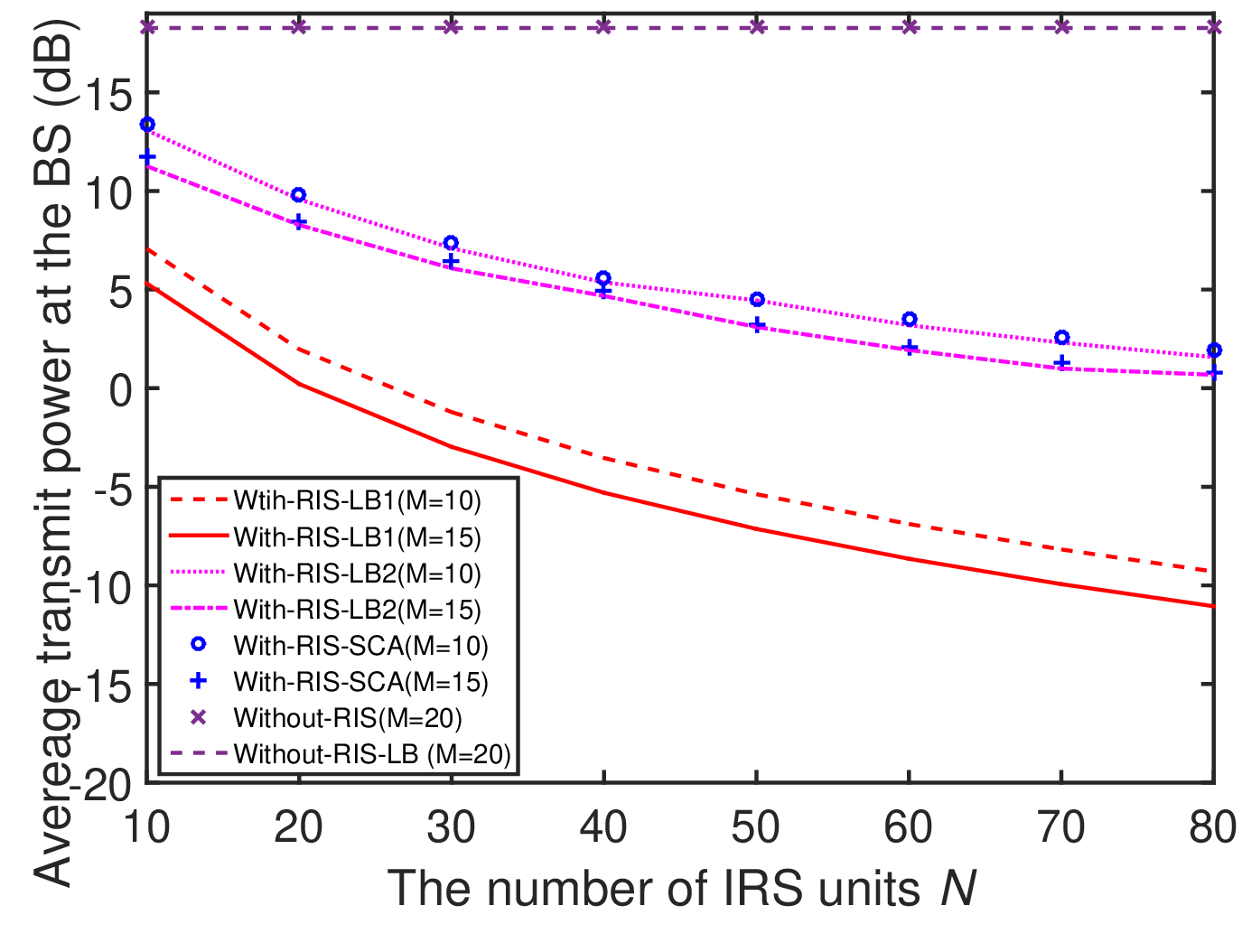}\\
  \caption{Average transmit power at the BS \\ versus   the  number of RIS units $N$  for \\ single-ME case.}
  \label{transmit-power-vs-N-signalUE}
\end{minipage}%
\begin{minipage}{.5\textwidth}
  \centering
  \includegraphics[scale=0.3]{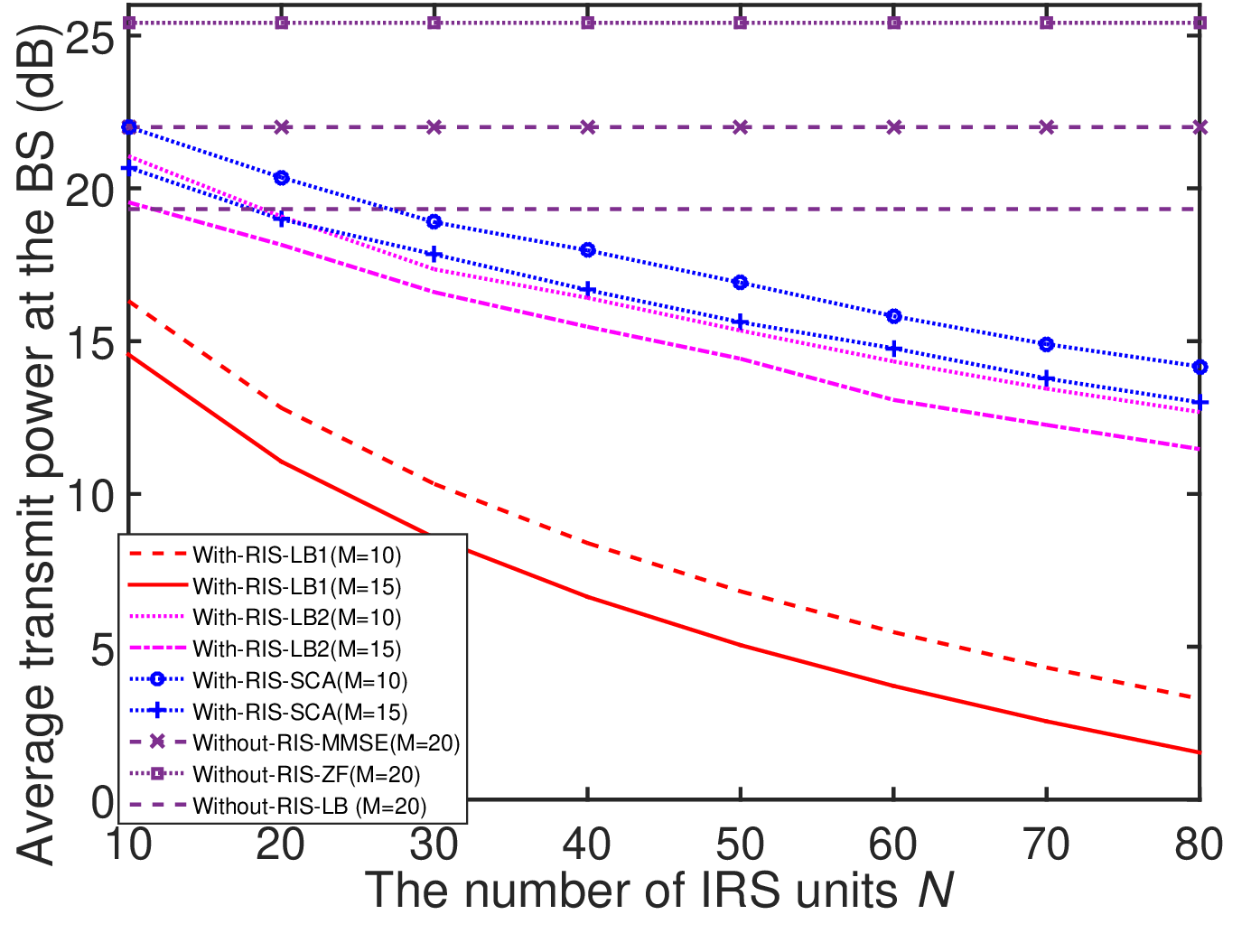}\\
  \caption{Average transmit power at the BS versus  the number of RIS units $N$  for multi-MEs case.
  }\label{transmit-power-vs-N-multiUE2}
\end{minipage}
\end{figure}

Figs.~\ref{transmit-power-vs-N-signalUE} and~\ref{transmit-power-vs-N-multiUE2} show how average transmit power at the BS changes with the number of RIS units $N$  under $\gamma=1$~dB~for $K=1$ and $K=5$, respectively. We can observe that, in  \mbox{ RIS aided}  systems,  with the increase of the number of RIS units $N$ and the number of antennas at the BS, the average transmit power decreases significantly. {\color{black} Furthermore, the simulation results  are  closer to \mbox{semi-analytical} lower bound, compared to  the  analytical lower bound which intuitively show  the transmit power scales with $1/N^2$ in the context of RIS communication systems.  The  reason  is given in Remark 1.}
 We can also observe that with the increase of the number of MEs,  the performance gap between   RIS aided systems and the  system without the RIS widens up gradually, which coincides with the
 trend obtained from the lower bounds. Interestingly, the average transmit power at the BS in RIS aided systems is much lower than the average transmit power without the RIS, even when the number of antennas at the BS in RIS aided systems~(i.e., $M=10,15$)~is less than that  without the RIS~(i.e., $M=20$). This means that \mbox{ RIS aided} systems with low-power consumption  elements   can increase the energy efficiency.
\begin{figure}
\setlength{\belowcaptionskip}{-0.4cm}
\setlength{\abovecaptionskip}{-0.3cm}
\centering
\begin{minipage}{.5\textwidth}
  \centering
  \includegraphics[scale=0.3]{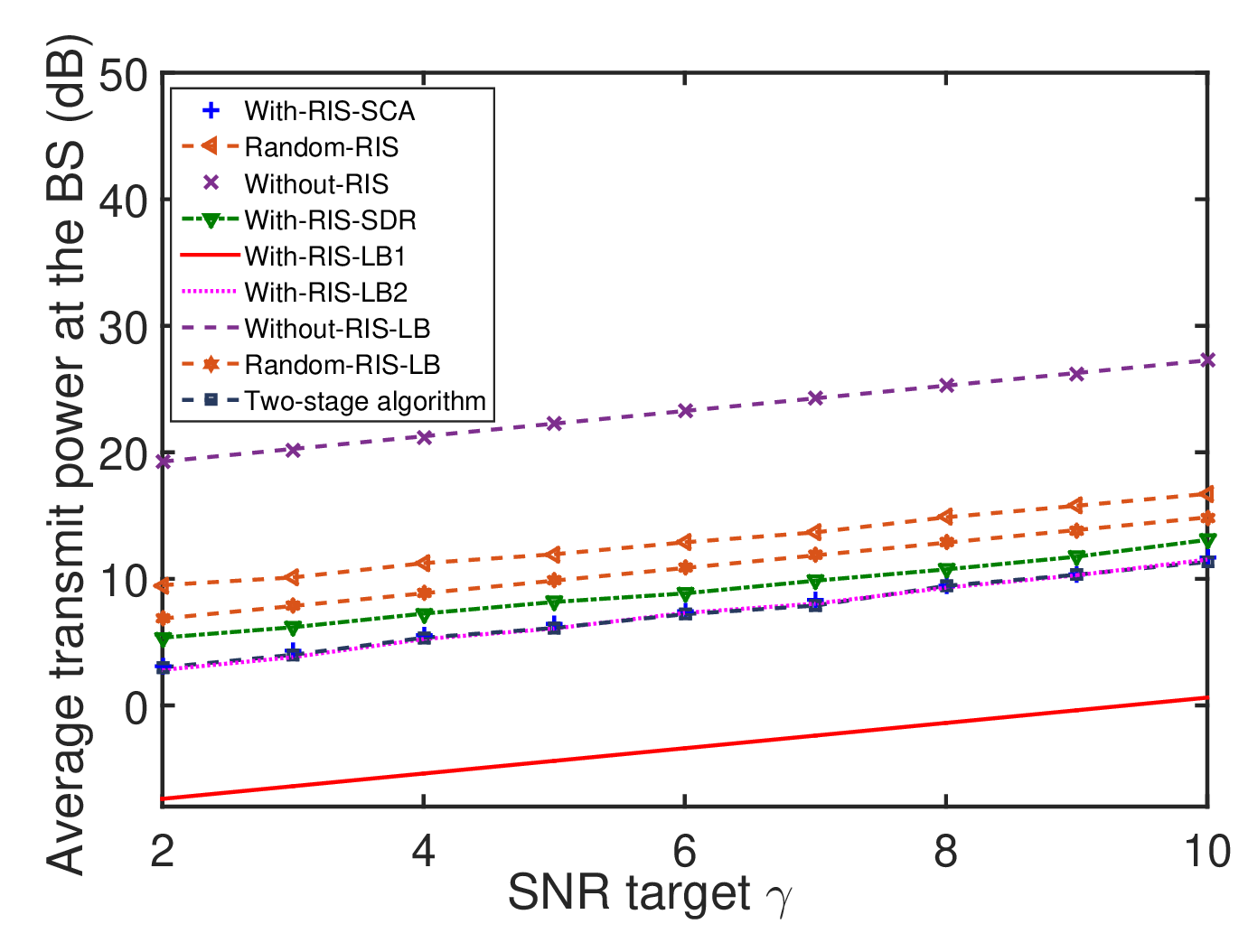}\\
  \caption{Average transmit power versus  \\  SNR target $\gamma$  for single-ME case.
  }\label{transmit-power-vs-gama-singleUE}
\end{minipage}%
\begin{minipage}{.5\textwidth}
  \centering
  \includegraphics[scale=0.3]{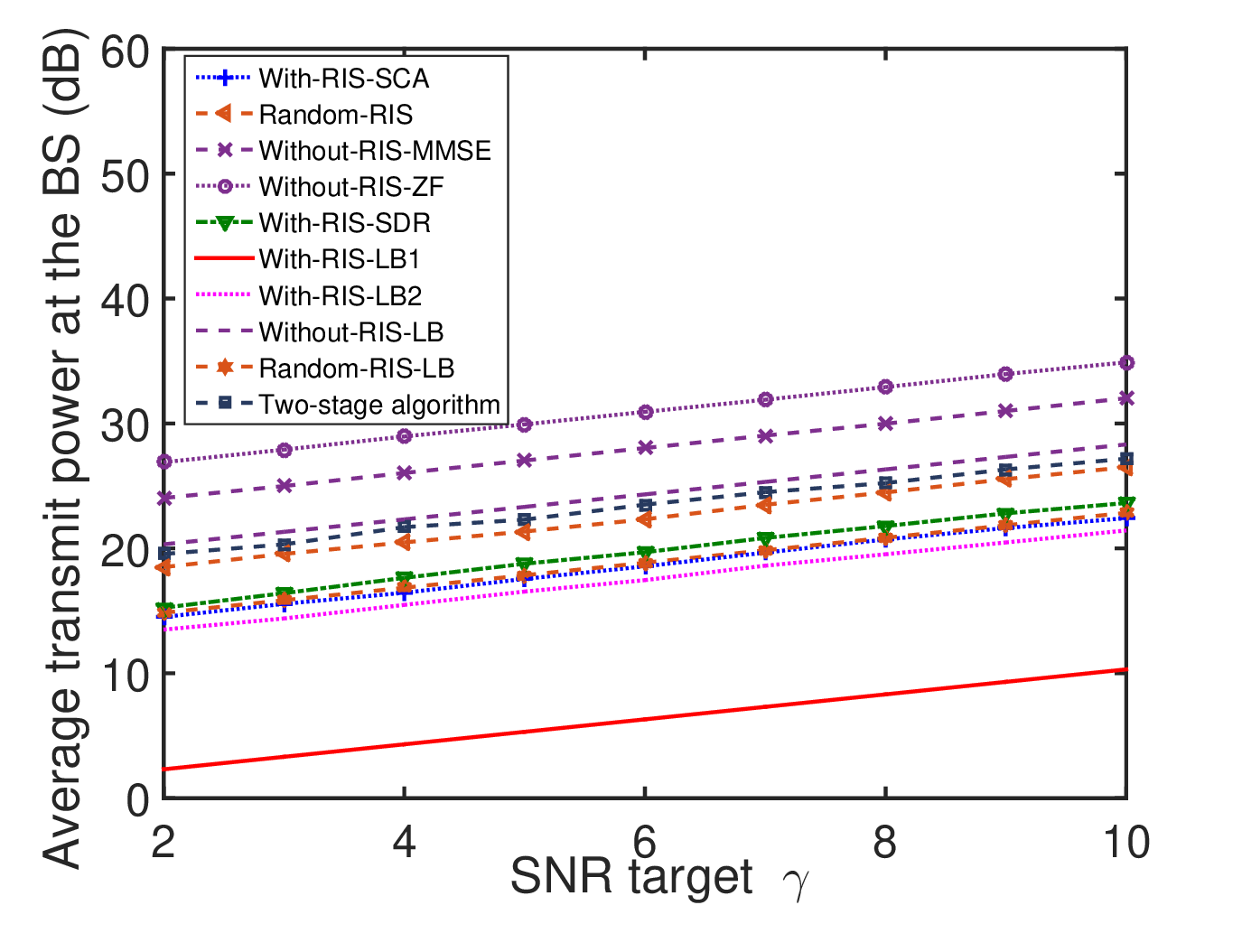}\\
  \caption{Average transmit power versus   SNR target   $\gamma$ for multi-MEs case.
  }\label{transmit-power-vs-gama-multiUE}
\end{minipage}
\end{figure}

\begin{figure}
\setlength{\belowcaptionskip}{-1cm}
\setlength{\abovecaptionskip}{-0.1cm}
  \centering
 \includegraphics[scale=0.33]{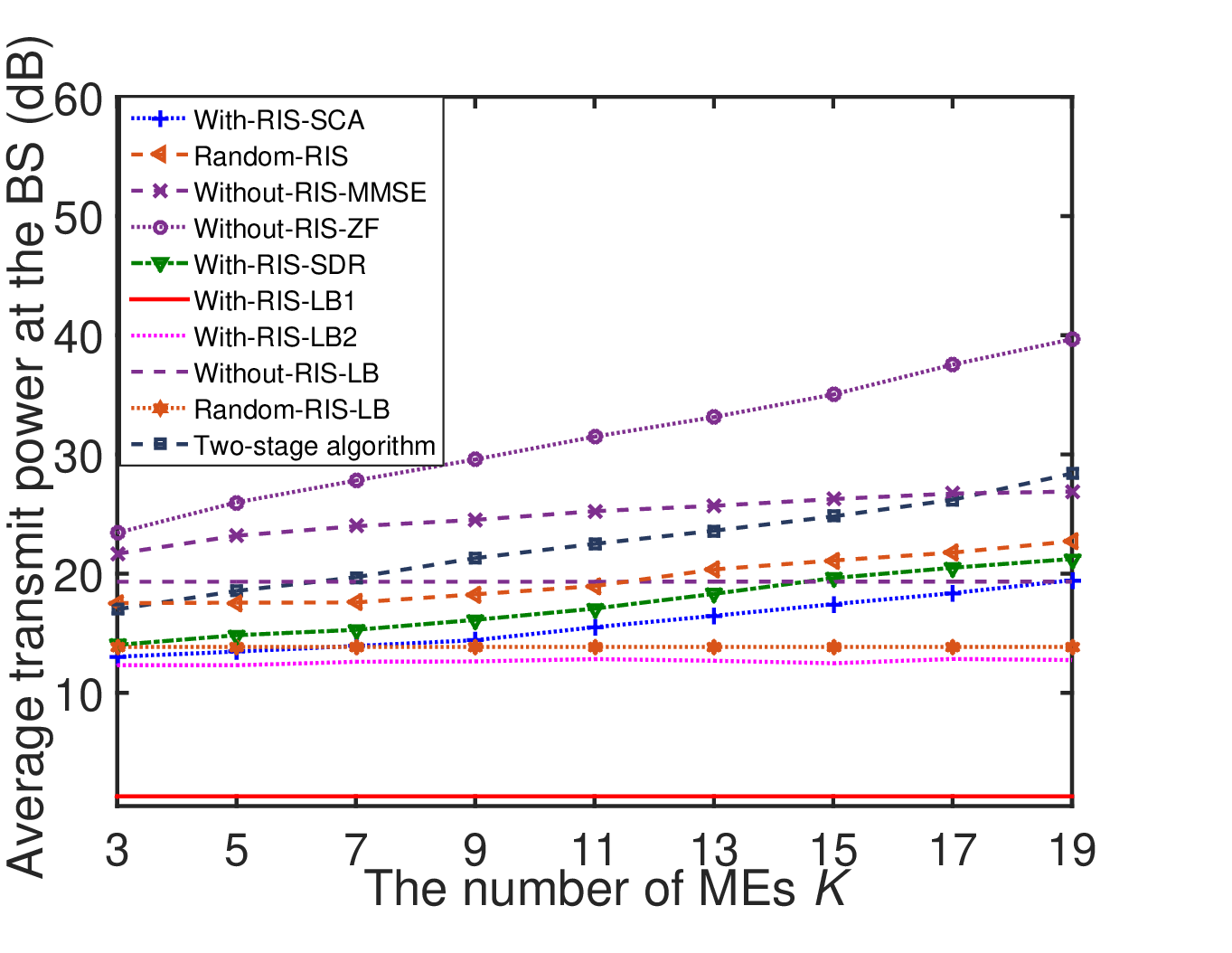}\\
  \caption{ Average transmit power versus  the number of MEs $K$. 
  }\label{transmit-power-vs-K}
\end{figure}

Figs.~\ref{transmit-power-vs-gama-singleUE} and~\ref{transmit-power-vs-gama-multiUE} show how average transmit power at the BS  changes with the SNR target $\gamma$ ranging from $2$ to $10$ under $M=20, N=70$ for $K=1$ and $K=5$, respectively. We can observe from the results that the average transmit power at the BS increases  almost linearly with the increase in the SNR target for single-ME and \mbox{multi-MEs} cases. We can also see that, when $\gamma$ ranges from $2$ to $10$, the  semi-analytical lower bound is  closer to the simulation results, compared to  the  \mbox{analytical} lower bound. 
{ \color{black} In addition, among these baselines, the two-stage algorithm with the  lower  complexity of $\mathcal{O}( (K+M^2)^{3.5}+ ((N+1)^2)^{3.5})$ achieves the same performance gain as our proposed alternating optimization algorithm based on SCA  and coincides with the  semi-analytical lower bound for  the single-ME case, while the transmit power at the BS of our proposed alternating optimization algorithms is significantly lower than that of the  two-stage algorithm for the multi-MEs case. The reason is that, the  performance of the two-stage algorithm  close to optimal. However, for the multi-MEs case, since the  phase shifts at the RIS are same for all MEs with different channel gain, the two-stage algorithm  cannot maximize the  combined channel power gain of different MEs  simultaneously~\cite{wu2018intelligentfull}.} 


Fig.~\ref{transmit-power-vs-K} shows how the average transmit power at the BS changes with the number of MEs for $M=20$, $\gamma=1$~dB, and $N=70$.  The results show that,  with the increase of the number of MEs, the average transmit power increases,  and is dramatically  lower than  the baselines. This means that \mbox{ RIS aided} systems   can  increase the energy efficiency.

\section{Conclusion} \label{sec:conclution}
In this paper, we have proposed  algorithms for power control problem at the BS with  QoS constraints in  RIS aided  wireless systems. Specifically, we have utilized  alternating optimization algorithms   to   jointly optimize  the transmit beamforming  at  the BS and the phase  shifts at  the RIS. Furthermore, we have derived  lower bounds for the average  transmit power.
Simulation results  have showed that   the average transmit power at the BS  is close to  the lower bound, and   is significantly lower  than that of communication systems without the RIS.

\bibliographystyle{IEEEtran}
\bibliography{ref}

\end{document}